\documentclass[a4paper,10pt]{article}
\usepackage[utf8]{inputenc}
 \usepackage{amsmath,amsfonts,amssymb}
 \usepackage{amsthm}
 \usepackage{mathtools}
 \usepackage{amsopn}
 \usepackage{tabularx,lipsum,environ}
 \usepackage{paralist}
 \usepackage{microtype}
 \usepackage{authblk}
 \usepackage{hyperref}
 \usepackage{fancyhdr}
 \usepackage{rotating}
 \usepackage{overpic}
 \usepackage{ucs}
 \usepackage{enumerate}
 \usepackage{graphicx}
 \usepackage{booktabs}
 \usepackage{varwidth}
 \usepackage{subfig}

\usepackage[color=green!20]{todonotes}

\newcommand{\cw}{\ensuremath\textrm{cw}}
\newcommand{\tw}{\ensuremath\textrm{tw}}
\newcommand{\pw}{\ensuremath\textrm{pw}}

\newcommand{\ms}{\ensuremath\textrm{ms}}

\newcommand{\dOr}{$d$-\textsc{Orientable Deletion}}
\newcommand{\cdOr}{\textsc{Capacitated}-$d$-\textsc{Orientable Deletion}}

\usepackage{comment}
\usepackage{tikz}
\tikzstyle{vertex}=[circle, draw, inner sep=1pt, minimum width=6pt]
\usetikzlibrary{decorations,decorations.pathmorphing,decorations.pathreplacing,fit,arrows,calc}

\renewcommand\leq\leqslant
\renewcommand\geq\geqslant
\renewcommand\le\leqslant
\renewcommand\ge\geqslant

\title{Parameterized Orientable Deletion\thanks{This work was financially supported by the ``PHC Sakura'' program (project GRAPA, number:  38593YJ), implemented by the French Ministry of Foreign Affairs, the French Ministry of Higher Education and Research and the Japan Society for Promotion of Science.
Y.O.\ was partially supported by JSPS KAKENHI grant numbers JP18K11168, JP18K11169, JP18H04091. M.L. and F.S. are partially supported by the project “ESIGMA” (ANR-17-CE40-0028)}}

\author[1]{Tesshu Hanaka}
\author[2,$\star$]{Ioannis Katsikarelis}
\author[2]{Michael Lampis}
\author[3]{Yota Otachi}
\author[2]{Florian Sikora}

\affil[1]{Chuo University, Tokyo, Japan}
\affil[2]{Universit\'e Paris-Dauphine, PSL Research University, CNRS, UMR, LAMSADE, Paris, France}
\affil[3]{Kumamoto University, Kumamoto, Japan}
\affil[$\star$]{Corresponding: i.katsikarelis@gmail.com}

\theoremstyle{plain}
\newtheorem{theorem}{Theorem}
\newtheorem{lemma}[theorem]{Lemma}
\newtheorem{corollary}[theorem]{Corollary}
\newtheorem{definition}[theorem]{Definition}

\newlength{\defbaselineskip}
\newcommand{\setlinespacing}[2]%
           {\setlength{\baselineskip}{#1 \defbaselineskip}}

\newlength{\btw}
\newlength{\stw}
\newsavebox\tmpbox

\date{}

\begin{document}

\maketitle

\begin{abstract}
A graph is $d$-orientable if its edges can be oriented so that the maximum
in-degree of the resulting digraph is at most $d$. $d$-orientability is a
well-studied concept with close connections to fundamental graph-theoretic
notions and applications as a load balancing problem. In this paper we consider
the \dOr\ problem: given a graph $G=(V,E)$, delete the minimum number of
vertices to make $G$ $d$-orientable.  We contribute a number of results that
improve the state of the art on this problem.  Specifically: 
\begin{itemize}
\item We show that the problem is W[2]-hard and $\log n$-inapproximable with
respect to $k$, the number of deleted vertices. This closes the gap in the
problem's approximability.
\item We completely characterize the parameterized complexity of the problem on
chordal graphs: it is FPT parameterized by $d+k$, but W[1]-hard by $d$ and W[2]-hard by $k$ alone.
\item We show that, under the SETH, for all $d,\epsilon$, the problem does not
admit a $O^*((d+2-\epsilon)^{\tw})$-time algorithm where $\tw$ is the graph's treewidth,
resolving as a special case an open problem on the complexity of
\textsc{PseudoForest Deletion}.
\item We show that the problem is W[1]-hard parameterized by the input graph's
clique-width.  Complementing this, we provide an algorithm running in time
$O^*(d^{O(d\cdot \cw)})$, showing that the problem is FPT by $d+\cw$, and improving
the previously best known algorithm for this case.
\end{itemize}
\end{abstract}


\section{Introduction}\label{sec:intro}

In this paper we study the following natural optimization problem: we are given
a graph $G=(V,E)$ and an integer $d$, and are asked to give directions to the edges of $E$
so that in the resulting digraph as many vertices as possible have in-degree  at most
$d$. Equivalently, we are looking for an orientation of $E$ such that the
set of vertices $K$ whose in-degree is strictly more than $d$ is minimized.
Such an orientation is called a $d$-orientation of $G[V\setminus K]$, and we say
that $K$ is a set whose deletion makes the graph $d$-orientable.
The problem of orienting the edges of an undirected graph so that the in-degree
of all, or most, vertices stays below a given threshold has been extensively
studied in the literature, in part because of its numerous applications. In
particular, one way to view this problem is as a form of scheduling, or load
balancing, where edges represent jobs and vertices represent machines. In this
case the in-degree represents the load of a machine in a given assignment, and
minimizing it is a natural objective (see e.g.
\cite{BateniCG09,ChakrabartyCK09,EbenlendrKS14,VerschaeW14}). Finding an
orientation where all in- or out-degrees are small is also of interest for the
design of efficient data structures \cite{ChrobakE91}. For more applications we refer the reader to
\cite{AsahiroJMO16,AsahiroJMOZ11,AsahiroMO11,AsahiroMOZ07,Bodlaender2017} and
the references therein.

\subparagraph*{State of the art.} $d$-orientability has been well-studied in
the literature, both because of its practical motivations explained above, but
also because it is a basic graph property that generalizes and is closely
related to fundamental concepts such as $d$-degeneracy (as a graph is
$d$-degenerate if and only if it admits an \emph{acyclic} $d$-orientation), and
bounded degree. This places \dOr\ in a general context of graph editing
problems that measure the distance of a given graph from having one of these
properties \cite{BetzlerBNU12,Mathieson10}.

The problem of orienting the edges of an
unweighted graph such that the maximum out-degree of vertices is minimized is solvable in polynomial
time \cite{AsahiroMOZ07}, though the problem becomes APX-hard
\cite{EbenlendrKS14} and even W[1]-hard parameterized by treewidth
\cite{Szeider11} if one allows edge weights. In this paper we focus on
unweighted graphs, for which computing the minimum number of vertices that need
to be deleted to make a graph $d$-orientable is easily seen to be NP-hard, as
the case $d=0$ corresponds to \textsc{Vertex Cover}. This hardness has
motivated the study of both polynomial-time approximation and parameterized
algorithms, as well as algorithms for specific graph classes. For
approximation, if the objective function is to maximize the number of
non-deleted vertices, the problem is known to be
$n^{1-\epsilon}$-inapproximable; if one seeks to minimize the number of deleted
vertices, the problem admits an $O(\log d)$-approximation, but it is not known
if this can be improved to a constant \cite{AsahiroJMO16}.  From the
parameterized point of view, the problem is W[1]-hard for any fixed $d$ if the
parameter is the number of non-deleted vertices, but FPT parameterized by $d$ and $\tw$ \cite{Bodlaender2017}. To the
best of our knowledge, the complexity of this problem parameterized by the
number of deleted vertices is open.
Moreover, \cite{BetzlerBNU12} shows that the related problem of deleting as few vertices as possible from a given graph, such that the resulting graph has maximum vertex degree $d$ is W[1]-hard parameterized by $\tw$, while FPT when parameterized by $\tw$ and $k$.

We remark that sometimes in the literature a $d$-orientation is an orientation
where all \emph{out-degrees} are at most $d$, but this can be seen to be
equivalent to our formulation by reversing the direction of all edges. \dOr\
has sometimes been called
\textsc{Min-}$(d+1)$\textsc{-Heavy}/\textsc{Max-}$d$\textsc{-Light}
\cite{AsahiroJMO16}, depending on whether one seeks to minimize the number of
deleted vertices, or maximize the number of non-deleted vertices (the two are
equivalent in the context of exact algorithms). The problem of finding an
orientation minimizing the maximum out-degree has also been called
\textsc{Minimum Maximum Out-degree} \cite{AsahiroMOZ07}.

An important special case that has recently attracted attention from the FPT
algorithms point of view is that of $d=1$. $1$-orientable graphs are called
pseudo-forests, as they are exactly the graphs where each component contains at
most one cycle. $1$-\textsc{Orientable Deletion}, also known as
\textsc{PseudoForest Deletion}, has been shown to admit a $3^k$ algorithm,
where $k$ is the number of vertices to be deleted
\cite{BodlaenderOO16,PhilipRS18}.

\subparagraph*{Our contribution.} We study the complexity of
\dOr\ mostly from the point of view of exact FPT algorithms. We contribute a
number of new results that improve the state of the art and, in some cases,
resolve open problems from the literature.

We first consider the parameterized complexity of the problem with respect to
the natural parameter $k$, the number of vertices to be deleted to make the
graph $d$-orientable. We show that for any fixed $d\ge 2$, \dOr\ is W[2]-hard
parameterized by $k$. This result is tight in two respects: it shows that,
under the ETH, the trivial $n^k$ algorithm that tries all possible solutions is
essentially optimal; and it cannot be extended to the case $d=1$, as in this
case the problem is FPT \cite{BodlaenderOO16}. Because our proof is a reduction
from \textsc{Dominating Set} that preserves the optimal, we also show that the
problem cannot be approximated with a factor better than $\ln n$. This matches
the performance of the algorithm given in \cite{AsahiroJMO16}, and closes a gap
in the status of this problem, as the previously best known hardness of
approximation bound was $1.36$ \cite{AsahiroJMO16}.

Second, we consider the complexity of \dOr\ when restricted to chordal graphs,
motivated by the work of \cite{Bodlaender2017}, who study the problem on
classes of graphs with polynomially many minimal separators. We are able to
completely characterize the complexity of the problem for this class of graphs
with respect to the two main natural parameters $d$ and $k$: the problem is
W[1]-hard parameterized by $d$, W[2]-hard parameterized by $k$, but solvable in
time $O^*(d^{O(d+k)})$, and hence FPT when parameterized by $d+k$. We recall
that the problem is poly-time solvable on chordal graphs when $d$ is a constant
\cite{Bodlaender2017}, and trivially in P in general graphs when $k$ is a
constant, so these results are in a sense tight.

Third, we consider the complexity of \dOr\ parameterized by the input graph's
treewidth, perhaps the most widely studied graph parameter. Our main
contribution here is a lower bound which, assuming the Strong ETH, states that
the problem cannot be solved in time less than $O^*((d+2)^{\tw})$, for any constant
$d\ge 1$. As a consequence, this shows that the $O^*(3^{\tw})$ algorithm given for
\textsc{PseudoForest Deletion} in \cite{BodlaenderOO16} is optimal under the
SETH.  We recall that Bodlaender et al. \cite{BodlaenderOO16} had explicitly
posed the existence of a better treewidth-based algorithm as an open problem;
our results settle this question in the negative, assuming the SETH. Our result
also extends the lower bound of \cite{LokshtanovMS18} which showed that
\textsc{Vertex Cover} (which corresponds to $d=0$) cannot be solved in
$O^*((2-\epsilon)^{\tw})$.

Finally, we consider the complexity of the problem parameterized by
clique-width. We recall that clique-width is probably the second most widely
studied graph parameter in FPT algorithms (after treewidth), so after having
settled the complexity of \dOr\ with respect to treewidth, investigating
clique-width is a natural question.  On the positive side, we present a dynamic
programming algorithm whose complexity is roughly $d^{O(d\cdot\cw)}$, and is
therefore FPT when parameterized by $d+\cw$.  This significantly improves upon
the dynamic programming algorithm for this case given in \cite{Bodlaender2017},
which runs in time roughly $n^{O(d\cdot\cw)}$.  The main new idea of this
algorithm, leading to its improved performance, is the observation that
sufficiently large entries of the DP table can be merged using a more careful
characterization of feasible solutions that involve large bi-cliques.  On the
negative side, we present a reduction showing that \dOr\ is W[1]-hard if $\cw$
is the only parameter and $d$ is part of the input.  This presents an interesting contrast with the case of
treewidth: for both parameters we can obtain algorithms whose running time is a
function of $d$ and the width; however, because graphs of treewidth $w$ always
admit a $w$-orientation (since they are $w$-degenerate), this immediately also
shows that the problem is FPT for treewidth, while our results imply that
obtaining a similar result for clique-width is impossible (under standard
assumptions).


\section{Definitions and Preliminaries}\label{sec:defs}

\subparagraph*{Complexity background.} We assume that the reader is familiar
with the basic definitions of parameterized complexity, such as the classes FPT
and W[1] \cite{CyganFKLMPPS15}. We will also make use of (slightly weaker statements of) the \emph{Exponential
Time Hypothesis} (ETH) and its \emph{strong} variant (SETH), two conjectures by Impagliazzo et al.\ asserting that
there is no $2^{o(n)}$-time algorithm for \textsc{3-SAT} on instances with $n$
variables \cite{Impagliazzo2001} (ETH) and that \textsc{SAT} cannot
be solved in time $O^*((2-\epsilon)^n)$ for any
$\epsilon>0$~\cite{Impagliazzo2001} (SETH).

\subparagraph*{Graph widths.} We also make use of standard graph width
measures, such as pathwidth, treewidth, and clique-width, denoted as $\pw, \tw,
\cw$ respectively.

A \emph{tree decomposition} of a graph $G=(V,E)$ is a pair $(\mathcal{X},T)$ with $T=(I,F)$ a tree and $\mathcal{X}=\{X_i|i\in I\}$ a family of subsets of $V$ (called \emph{bags}), one for each node of $T$, with the following properties:
 \begin{enumerate}[1)]
  \item $\bigcup_{i\in I}X_i=V$;
  \item for all edges $(v,w)\in E$, there exists an $i\in I$ with $v,w\in X_i$;
  \item for all $i,j,k\in I$, if $j$ is on the path from $i$ to $k$ in $T$, then $X_i\cap X_k\subseteq X_j$.
 \end{enumerate}
 The \emph{width} of a tree decomposition $((I,F),\{X_i|i\in I\})$ is $\max_{i\in I}|X_i|-1$. The \emph{treewidth} of a graph $G$ is the minimum width over all tree decompositions of $G$, denoted by $\textrm{tw}(G)$ (\cite{BodlaenderK08,Bodlaender06,Bodlaender00}). The definition of \emph{pathwidth} is similar, with $T$ being a path instead of a tree in this case.

We next define the notion of \emph{clique-width} (see \cite{CourcelleMR00}): the set of graphs of cliquewidth $\textrm{cw}$ is the set of vertex-labelled graphs that can be inductively constructed by using the following operations:
\begin{enumerate}[1)]
  \item Introduce: $i(l)$, for $l\in[1,\textrm{cw}]$ is the graph consisting of a single vertex with label $l$;
  \item Join: $\eta(G,a,b)$, for $G$ having cliquewidth $\textrm{cw}$ and $a,b\in[1,\textrm{cw}]$ is the graph obtained from $G$ by adding all possible edges between vertices of label $a$ and vertices of label $b$;
  \item Rename: $\rho(G,a,b)$, for $G$ having cliquewidth $\textrm{cw}$ and $a,b\in[1,\textrm{cw}]$ is the graph obtained from $G$ by changing the label of all vertices of label $a$ to $b$;
  \item Union: $G_1\cup G_2$, for $G_1,G_2$ having cliquewidth $\textrm{cw}$ is the disjoint union of graphs $G_1,G_2$.
 \end{enumerate} Note we here assume the labels are integers in $[1,\textrm{cw}]$, for ease of exposition.
 
A \emph{clique-width expression} of width $\textrm{cw}$ for $G=(V,E)$ is a recipe for constructing a $\textrm{cw}$-labelled graph isomorphic to $G$. More formally, a cliquewidth expression is a rooted binary tree $T_G$, such that each node $t\in T_G$ has one of four possible types, corresponding to the operations given above. In addition, all leaves are introduce nodes, each introduce node has a label associated with it and each join or rename node has two labels associated with it. For each node $t$, the graph $G_t$ is defined as the graph obtained by applying the operation of node $t$ to the graph (or graphs) associated with its child (or children). All graphs $G_t$ are subgraphs of $G$ and for all leaves of label $l$, their associated graph is $i(l)$.

Additionally, we will require the equivalent definition of pathwidth via the \emph{mixed search number} $\ms(G)$ (see \cite{TAKAHASHI1995253}). In a \emph{mixed search game}, a graph $G$ is considered as a system of tunnels. Initially, all edges are contaminated by a gas and an edge is \emph{cleared} by placing searchers at both its endpoints simultaneously or by sliding a searcher along the edge. A cleared edge is re-contaminated if there is a path from a contaminated edge to the cleared edge without any searchers on its vertices or edges. A search is a sequence of operations that can be of the following types: (a) placement of a new searcher on a vertex; (b) removal of a searcher from a vertex; (c) sliding a searcher on a vertex along an incident edge and placing the searcher on the other end. A search strategy is winning if after its termination all edges are cleared. The mixed search number of $G$, denoted by $\ms(G)$, is the minimum number of searchers required for a winning strategy of mixed searching on $G$.

The following lemma from \cite{TAKAHASHI1995253} shows the relationship between $\ms(G)$ and $\pw(G)$:

\begin{lemma}\label{lem:mixed_search} \cite{TAKAHASHI1995253} For a graph $G$, it is $\pw(G)\le\ms(G)\le\pw(G)+1$. \end{lemma}

We also refer the reader to the standard textbooks \cite{CyganFKLMPPS15,CE12} and recall the following well-known
relations:

\begin{lemma}\label{lem:widths}

For all graphs $G=(V,E)$ we have $\tw(G)\le \pw(G)$ and $\cw(G)\le
\pw(G)+2$.

\end{lemma}

\subparagraph*{Graphs and Orientability.} We use standard graph-theoretic
notation.  If $G=(V,E)$ is a graph and $S \subseteq V$, $G[S]$ denotes the
subgraph of $G$ induced by $S$.  For $v \in V$, the set of neighbors of $v$ in
$G$ is denoted by $N_G(v)$, or simply $N(v)$, and $N_G(S):=(\bigcup_{v \in
S}N(v)) \setminus S$ will often be written just $N(S)$.  We define $N[v]:=N(v)
\cup \{v\}$ and $N[S]:=N(S) \cup S$. Depending on the context, we use $(u,v)$,
where $u,v \in V$ to denote either an undirected edge connecting two vertices
$u,v$, or an arc (that is, a directed edge) with tail $u$ and head $v$. An
\emph{orientation} of an undirected graph $G=(V,E)$ is a directed graph on the
same set of vertices obtained by replacing each undirected edge $(u,v)\in E$
with either the arc $(u,v)$ or the arc $(v,u)$. 
In a directed graph we define the in-degree $\delta^-(u)$ of a vertex $u$ as
the number of arcs whose head is $u$.  A $d$-orientation of a graph $G=(V,E)$
is an orientation of $G$ such that all vertices have in-degree at most $d$. If
such an orientation exists, we say that $G$ is $d$-orientable.  Deciding if a
given graph is $d$-orientable is solvable in polynomial time, even if $d$ is
part of the input \cite{AsahiroMOZ07}.  Let us first make some easy
observations on the $d$-orientability of some basic graphs.

\begin{lemma}\label{lem:clique}

$K_{2d+1}$, the clique on $2d+1$ vertices, is $d$-orientable. Furthermore, in
any $d$-orientation of $K_{2d+1}$ all vertices have in-degree $d$.

\end{lemma}
\begin{proof}

We show that there is a way to orient all the edges of any $K_{2d+1}$ so that
all vertices have in-degree exactly $d$.  To see this, number the vertices
$\{0,1,\ldots,2d\}$, and then for each $i\in\{0,\ldots,2d\}$ orient away from
vertex $i$ all edges whose other endpoint is $\{i+1,\ldots,i+d\}$, where
addition is done modulo $2d+1$. Observe that this defines the orientation of
all edges, and for each vertex it orients away from it $d$ of its $2d$ incident
edges. Hence, all vertices have in-degree $d$ in the end (see also Figure~\ref{fig:cliquegadget}).

For the second part, observe that $K_{2d+1}$ has $d(2d+1)$ edges and $2d+1$
vertices, hence in any orientation the average in-degree must be exactly $d$.
In an orientation where the maximum in-degree is $d$ we therefore also have
that the minimum in-degree is also $d$.  \end{proof}

\begin{lemma}\label{biclique_nonorientable}
 The complete bipartite graph $K_{2d+1,2d}$ is not $d$-orientable.
\end{lemma}
\begin{proof}
 We observe that in any orientation the sum of all in-degrees is equal to the
total number of edges, which is $4d^2+2d$.  Hence, the average in-degree is
$>d$, meaning there will always be at least one vertex of in-degree $>d$.
\end{proof}

\begin{definition}

In \dOr\ we are given as input a graph $G=(V,E)$ and an integer $d$. We are
asked to determine the smallest set of vertices $K\subseteq V$ (the \emph{deletion set}) such that
$G[V\setminus K]$ admits a $d$-orientation.

\end{definition}
\begin{definition}

In \cdOr\ we are given as input a graph $G=(V,E)$, an integer $d\ge 1$, and a
\emph{capacity} function $\mathbf{c}: V\to\{0,\ldots,d\}$. We are asked to determine
the smallest set of vertices $K\subseteq V$ such that $G[V\setminus K]$ admits
an orientation with the property that for all $u\in V\setminus K$, the
in-degree of $u$ is at most $\mathbf{c}(u)$.

\end{definition}
It is clear that \cdOr\ generalizes \dOr, which corresponds to the case where
we have $\mathbf{c}(u)=d$ for all vertices. It is, however, not hard to see that the two
problems are in fact equivalent, as shown in the following lemma.  Furthermore,
the following lemma shows that increasing $d$ can only make the problem harder.
\begin{lemma}\label{lem:saturation}

There exists a polynomial-time algorithm which, given an instance
$[G=(V,E),d,\mathbf{c}]$ of \cdOr, and an integer $d'\ge d$, produces an
equivalent instance $[G'=(V',E'),d']$ of \textsc{$d'$-Orientable Deletion}, with the same optimal value and
the following properties: $\pw(G')\le \pw(G)+2d'+1$, $\cw(G')\le \cw(G)+4$, and
if $G$ is chordal then $G'$ is chordal.

\end{lemma}


\begin{figure}
\centering
\begin{minipage}{.4\textwidth}
\centering
\begin{tikzpicture}[scale=0.8,transform shape]
\begin{scope}[-latex]

\node[vertex] (0) at (0,0.3) {0};
\node[vertex] (1) at (-1,-0.5)  {1};
\node[vertex] (2) at (1,-0.5)  {2};
\node[vertex] (3) at (-0.5,-1.5)  {3};
\node[vertex] (4) at (0.5,-1.5)  {4};

\draw (0) -- (1);
\draw (0) -- (2);
\draw (1) -- (2);
\draw (1) -- (3);
\draw (2) to node[below]{} (3);
\draw (2) -- (4);
\draw (3) -- (4);
\draw (3) -- (0);
\draw (4) -- (0);
\draw (4) -- (1);

\node[vertex] (u) at (-2.5,-0.5) {$u$};
\draw (1) -- (u);

\end{scope}
\end{tikzpicture}
\caption{A $2$-orientation of a clique $K_5$. Observe that any edge connecting a vertex $u$ to the clique must be oriented towards $u$ (setting its \emph{capacity}) to maintain a $2$-orientation.}
\label{fig:cliquegadget}
\end{minipage}%
\hspace{0.5cm}
\begin{minipage}{.5\textwidth}
\centering

\begin{tikzpicture}[scale=0.8,transform shape]
\node[vertex] (d1) at (0,0) {$v_1$};
\node[vertex,white,fill=black] (d2) at (-0.5,-0.5)  {$v_2$};
\node[vertex] (d3) at (0.5,-0.5)  {$v_3$};
\node[vertex] (d4) at (-0.5,-1.5)  {$v_4$};
\node[vertex,white,fill=black] (d5) at (0.5,-1.5)  {$v_5$};

\draw (d3) -- (d1)--(d2) --(d3)--(d5) -- (d4) -- (d2);

\begin{scope}[xshift=1.5cm,-latex]
\foreach \x in {1,...,5}
{
\node[vertex] (v2\x) at (0.8*\x,-2) {$v_{\x}$};
}
\foreach \x in {1,3,4}
{
\node[vertex] (v1\x) at (0.8*\x,0) {$v_{\x}$};
}

\node[vertex,white,fill=gray] (v12) at (0.8*2,0) {$v_{2}$};
\node[vertex,white,fill=gray] (v15) at (0.8*5,0) {$v_{5}$};

\node () at (0.2,0) {$V_1$};
\node () at (0.2,-2) {$V_2$};

\node () at (5,0) {$\mathbf{c} = 0$};
\node () at (5,-2) {$\mathbf{c} = 1$};

\draw [-,decorate,decoration={brace,amplitude=3pt}] (4.3,-0.3) -- (4.3,-1.7) node [black,midway,xshift=0.65cm]  {$\mathbf{c} = 2$};

\node[vertex] at (0.8,-0.7) (t12)  {}; 

\draw (v21) -- (t12);
\draw (v11) -- (t12);
\draw (t12) edge[dotted,-] (v12);
\draw (v13) -- (v21);

\node[vertex] at (1.3,-0.7) (t12-2)  {}; 
\draw (v22) -- (t12-2);
\draw (v11) -- (t12-2);
\draw (t12-2) edge[dotted,-] (v12);

\node[vertex] at (2.4,-0.7) (t34-2)  {}; 
\draw (t34-2) -- (v22);
\draw (v13) -- (t34-2);
\draw (v14) -- (t34-2);

\node[vertex] at (2.4,-1.3) (t12-3)  {}; 
\draw (v23) -- (t12-3);
\draw (v11) -- (t12-3);
\draw (t12-3) edge[dotted,-] (v12);	

\node[vertex] at (3,-0.7) (t35-3)  {}; 
\draw (v13) -- (t35-3);
\draw (v15) edge[dotted,-] (t35-3);
\draw (v23) -- (t35-3);

\node[vertex] at (3.2,-1.3) (t45-4)  {}; 
\draw (v12) edge[dotted,-] (v24);
\draw (v14) -- (t45-4);
\draw (v15) edge[dotted,-] (t45-4);
\draw (v24) -- (t45-4);

\node[vertex] at (4,-0.7) (t45-5)  {}; 
\draw (v14) -- (t45-5);
\draw (v15) edge[dotted,-] (t45-5);
\draw (v25) -- (t45-5);
\draw (v13) edge[bend left=15] (v25);

\end{scope}

\end{tikzpicture}

\caption{{\bf Left:} A graph with its dominating set in black. {\bf Right:} The corresponding instance and 2-orientation, where deleted vertices are in gray. Original edges with a deleted vertex as an endpoint are dotted.
\label{WhardnessDS} }
\end{minipage}
\end{figure}

\begin{proof}

We use a saturation gadget to simulate the fact that some vertices of $G$ are
meant to be able to accept strictly fewer than $d$ incoming edges. In
particular, for each $u\in V$ for which $\mathbf{c}(u)<d'$ we construct a
clique $K_{2d'+1}$ on $2d'+1$ new vertices and connect $d'-\mathbf{c}(u)$ of these
vertices (arbitrarily chosen) to $u$. 

Let us argue that this produces an equivalent instance. First, consider a
solution to the original instance. We delete the same set of vertices from
$G'$, and use the same orientation for edges of $E$. For the added edges, there
is a way to orient all the edges of any $K_{2d'+1}$ by Lemma \ref{lem:clique}.
We therefore use this orientation for all the copies of $K_{2d'+1}$ we added,
and orient all other edges incident on such cliques away from the clique. This
does not make the in-degree of any vertex of $V$ higher than $d'$, since such
vertices have at most $\mathbf{c}(u)$ incoming edges from $E$ (by assumption),
and $d'-\mathbf{c}(u)$ edges coming from the clique.

For the converse direction, consider a solution of the new instance. We first
observe that without loss of generality we may assume that the solution does
not delete any of the vertices of the cliques we added to the graph. This is
because, if the solution deletes a vertex of a $K_{2d'+1}$ attached to $u\in
V$, we can instead delete $u$ and use the orientation described above for the
vertices of the clique (which is now disconnected from the graph). If the
solution does not delete any vertex from a $K_{2d'+1}$, by Lemma
\ref{lem:clique}, all vertices of the clique have in-degree exactly $d'$ in the
clique. This implies that all edges with one endpoint in the clique must be
oriented away from the clique, hence any vertex $u\in V$ may have in-degree at
most $\mathbf{c}(u)$ from edges of $E$.

For the width bounds, first consider any path decomposition of $G$. We can
construct a path decomposition of $G'$ as follows: for each $u\in V$ to which
we attached a $K_{2d'+1}$ we find a bag $B$ of the original decomposition that
contains $u$, and insert after $B$ a copy of the same bag into which we add all
the vertices of the clique attached to $u$. Consider now a clique-width
expression for $G$. We can use it to construct a clique-width expression for
$G'$ using three new labels by introducing a clique of size $d'-\mathbf{c}(u)$ using two labels, connecting it to $u$ using another and then completing the clique using the first two labels. Finally, to see that $G'$ is chordal if $G$
is chordal observe that no induced cycle of length $4$ or more can contain any
of the new vertices since they induce a union of cliques and are connected to
the rest of the graph through single cut-vertices.  \end{proof}



\section{Hardness of Approximation and W[2]-hardness}\label{sec:approx_whard}

In this section we present a reduction from \textsc{Dominating Set} to \dOr\
for $d\ge2$ that exactly preserves the size of the solution. As a result, this
establishes that, for any fixed $d\ge 2$, \dOr\ is W[2]-hard, and the minimum
solution cannot be approximated with a better than logarithmic factor. We
observe that it is natural that our reduction only works for $d\ge 2$, as the
problem is known to be FPT for $d=1$, which is known as \textsc{PseudoForest
Deletion}, and $d=0$, which is equivalent to \textsc{Vertex Cover}.

\begin{theorem}\label{thm:whard} For any $d\ge 2$, \dOr\ is W[2]-hard
parameterized by the solution size $k$. Furthermore, for any $d\ge 2$, \dOr\
cannot be solved in time $f(k)\cdot n^{o(k)}$, unless the ETH is false. \end{theorem}
\begin{proof}

We will describe a reduction to \cdOr\ for $d=2$, from the well-known \textsc{Dominating Set} problem: we are given a graph $G=(V,E)$ and an integer $k$ and are asked if there exists a dominating set of size $k$. \textsc{Dominating Set} is
W[2]-hard and not solvable in $f(k)\cdot n^{o(k)}$ under the ETH \cite{CyganFKLMPPS15}. We
will then invoke Lemma \ref{lem:saturation} to obtain the claimed result for
\dOr.
Let $[G(V,E),k]$ be an instance of Dominating Set. We begin by constructing a
bipartite graph $H$ by taking two copies of $V$, call them $V_1,V_2$. For each
$v\in V_2$ we construct a binary tree with $|N_G[v]|$ leaves. We identify the
root of this binary tree with $v\in V_2$ and its leaves with the corresponding
vertices in $V_1$.  We now define the capacities of our vertices: each vertex
of $V_1$ has capacity $0$; each internal vertex of the binary trees has
capacity $2$; and each vertex of $V_2$ has capacity $1$.

We will now claim that $G$ has a dominating set of size $k$ if and only if $H$ can be oriented in a way that respects the capacities by deleting at
most $k$ vertices.

For the forward direction, suppose that there is a dominating set in $G$ of
size $k$. In $H$ we delete the corresponding vertices of $V_1$. We
argue that the remaining graph is orientable in a way that respects the
capacities.  We compute an orientation as follows: 

\begin{enumerate}
\item We orient the remaining
incident edges away from every vertex of $V_1$ that is not deleted.
\item For each non-leaf vertex $u$ of the binary tree rooted at $v\in V_2$ we
define the orientation of the edge connecting $u$ to its parent as follows: $u$
is an  ancestor of a set $S_u\subseteq N_G[v]$ of vertices of $V_1$. If $S_u$
contains a deleted vertex, then we orient the edge connecting $u$ to its parent
towards $u$, otherwise we orient it towards $u$'s parent.
\end{enumerate}

The above description completely defines the orientation of the remaining graph
(see also Figure~\ref{WhardnessDS}). Let us argue why the orientation respects
all capacities.  This should be clear for vertices of $V_1$.  For any non-leaf
vertex $u$ of a binary tree, if we orient the edge connecting it to its parent
away from $u$, then the in-degree of $u$ is at most $2$, which is its capacity.
On the other hand, if we orient this edge towards $u$, there is a deleted
vertex in $S_u$.  However, this implies either that one of $u$'s children has
been deleted, or that one of the edges connecting $u$ to one of its children is
oriented away from $u$. In both cases, the in-degree of $u$ is at most $2$, equal to its
capacity. Finally, for each $u\in V_2$, if we started with a dominating set,
then one of the children of $u$ in the binary tree is either deleted or its
edge to $u$ is oriented towards it. 

For the converse direction, suppose that there is a set of $k$ vertices in $H$ whose deletion makes the graph orientable in a way that respects the
capacities.  Suppose now that we have a solution that deletes some vertex $v\in
V_2$ or some internal vertex of a binary tree. We re-introduce $v$ in the
graph, orient all its incident edges towards $v$, and then delete one of the
children of $v$. This preserves the size and validity of the solution.
Repeating this argument ends with a solution that only deletes vertices of
$V_1$. We now claim that these $k$ vertices are a dominating set.  To see this,
observe that any undeleted vertex of $V_1$ has all its edges connecting it to
binary trees oriented away from it.  Hence, if there is a binary tree with root
$v\in V_2$ such that none of its leaves are undeleted, all its internal edges
must be oriented towards $v$, which would make the in-degree of $v$ greater
than its capacity.  \end{proof}

\begin{corollary}
\label{cor:approx}
For any $d\ge 2$, if there exists a polynomial-time $o(\log n)$-approximation
for \dOr, then P=NP.

\end{corollary}
\begin{proof}

We observe that the reduction from \textsc{Dominating Set} given in Theorem~\ref{thm:whard} exactly preserves the optimal value. We can therefore invoke
known hardness of approximation results on Dominating Set (see e.g.\
\cite{Moshkovitz15}).\end{proof}


\section{Chordal Graphs}\label{sec:chordal}

In this section we consider the complexity of \dOr\ on chordal graphs
parameterized by either $d$ or $k$ (the number of deleted vertices). Our main
results state that the problem is W[1]-hard for each of these parameters
individually (Theorems \ref{thm:whard-chordal} and \ref{thm:whard-chordal2});
however, the problem is FPT parameterized by $d+k$ (Theorem
\ref{thm:chordal-alg}).

\begin{theorem}\label{thm:whard-chordal}

\dOr\ is W[1]-hard on chordal graphs parameterized by $d$. Furthermore, it cannot be solved in time $n^{o(d)}$, under the ETH.

\end{theorem}
\begin{proof}

We give a reduction from \textsc{Independent Set}: given a graph $G=(V,E)$ and an integer $k$, we
are asked if $G$ contains an independent set of size $k$. We assume without loss
of generality that $k$ is odd. We set $d=k$ and we will construct a chordal
instance of \cdOr, on which we will invoke Lemma \ref{lem:saturation} to obtain
the claimed result, together with the standard fact that \textsc{Independent
Set} is W[1]-hard and not solvable in $n^{o(k)}$ under the ETH \cite{CyganFKLMPPS15}.

We construct a graph $G'$ from $G$ as follows.  First, we subdivide all edges
of $E$, and then we connect all (original) vertices of $V$ into a clique.  Let
us now define the capacities: for any vertex introduced during a subdivision of
an edge $e\in E$ we set its capacity to 1.  Furthermore, for each $u\in V$ we
set its capacity to $\frac{d-1}{2}$.  This completes the construction (see also
Figure~\ref{chordalIS}), while our budget will be set to $n-k$. Observe that we have constructed a split graph,
therefore $G'$ is chordal.

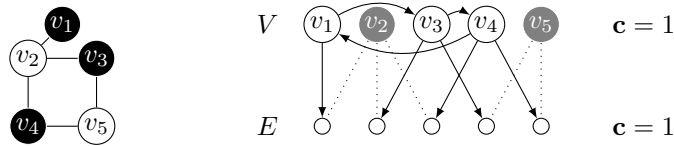
\begin{figure}
\centering

\begin{tikzpicture}[scale=0.9]
\node[vertex,white,fill=black] (d1) at (0,0) {$v_1$};
\node[vertex] (d2) at (-0.5,-0.5)  {$v_2$};
\node[vertex,white,fill=black] (d3) at (0.5,-0.5)  {$v_3$};
\node[vertex,white,fill=black] (d4) at (-0.5,-1.5)  {$v_4$};
\node[vertex] (d5) at (0.5,-1.5)  {$v_5$};

\draw (d1)--(d2) --(d3)--(d5) -- (d4) -- (d2);

\begin{scope}[xshift=3cm,-latex]

\foreach \x in {1,...,5}
{
\node[vertex] (e\x) at (0.8*\x,-1.5) {};
}
\foreach \x in {1,3,4}
{
\node[vertex] (v\x) at (0.8*\x,0) {$v_{\x}$};
}

\node[vertex,white,fill=gray] (v2) at (0.8*2,0) {$v_{2}$};
\node[vertex,white,fill=gray] (v5) at (0.8*5,0) {$v_{5}$};

\node () at (0.8*0,0) {$V$};
\node () at (0.8*0,-1.5) {$E$};
\node () at (5.5,0) {$\mathbf{c}=1$};
\node () at (5.5,-1.5) {$\mathbf{c}=1$};

\draw (v1) -- (e1);
\draw (v2) edge[dotted,-] (e1);

\draw (v3) -- (e2);
\draw (v2) edge[dotted,-] (e2);

\draw (v4) -- (e3);
\draw (v2) edge[dotted,-] (e3);

\draw (v3) -- (e4);
\draw (v5) edge[dotted,-] (e4);

\draw (v4) -- (e5);
\draw (v5) edge[dotted,-] (e5);

\draw (v1) edge[bend left] (v3);
\draw (v3) edge[bend left] (v4);
\draw (v4) edge[bend left] (v1);

\end{scope}

\end{tikzpicture}
\caption{A sample graph with its independent set of size 3 in black on the left. The corresponding constructed instance with its 3-orientation on the right. The $n-k=2$ deleted vertices are drawn in gray. Original edges with a deleted vertex as an endpoint are drawn dotted. 
For clarity, edges of the clique in $V$ for deleted vertices are not drawn. \label{chordalIS} }
\end{figure}

Suppose that there is an independent set of size $k$ in $G$, hence there is a
vertex cover of size $n-k$. We delete the corresponding $n-k$ vertices of $G'$ and claim that
the remaining graph is orientable in a way that respects the capacities.  For
every edge $e\in E$, its corresponding vertex has a deleted neighbor (since we
started by deleting the vertices corresponding to a vertex cover). Hence, it now has degree at most $1$, which is at
most equal to its capacity, so we orient its possible remaining incident edge
towards it.  Finally, for the $k=d$ undeleted vertices of $V$, we use Lemma
\ref{lem:clique} to obtain an orientation of the $K_d$ they induce where all
vertices have in-degree $\frac{d-1}{2}$, which is equal to their capacities.

For the converse direction, suppose that it is possible to orient the graph
respecting the capacities by deleting at most $n-k$ vertices. To simplify
things we assume we have a solution that deletes exactly $n-k$ vertices, which
can be achieved by adding arbitrary vertices to a smaller solution. If a
solution deletes a vertex corresponding to $e\in E$, we can place that vertex
back, orient all edges towards it, and if it now has in-degree $2$ arbitrarily
delete one of its neighbors in $V$. We therefore suppose that the solution
deletes $n-k$ vertices of $V$ and we must show that these vertices are a vertex
cover of $G$.  Suppose for contradiction that there is an edge $e\in E$ such
that neither of its endpoints in $V$ was deleted. We now observe that in any
orientation one edge connecting the vertex produced in the subdivision of $e$
to $V$ must be oriented towards $V$, because the vertex corresponding to $e$
has capacity $1$. However, the $d$ undeleted vertices of $V$ form a clique, and
by Lemma \ref{lem:clique}, any orientation of the edges of this clique that
gives all vertices in-degree at most $\frac{d-1}{2}$ (their capacities), gives
all vertices exactly this in-degree. Hence, the additional edge from $e$ will
force one non-deleted vertex to violate its capacity.  \end{proof}

\begin{theorem}\label{thm:whard-chordal2}
\dOr\ is W[2]-hard on chordal graphs parameterized by
the solution size $k$ and cannot be solved in time
$n^{o(k)}$ under the ETH, when $d$ is part of the input.
\end{theorem}
   \begin{proof} 
   We start from an instance of \textsc{Dominating Set}: we are given a graph $G=(V,E)$ and
   an integer $k$ and are asked if there exists a dominating set of size $k$. We
   will retain the same value of $k$ and construct a chordal instance of \cdOr,
   for which we later invoke Lemma \ref{lem:saturation}.  Let $|V|=n$ and we
   assume without loss of generality that $n-k$ is odd (otherwise we can add an
  isolated pair of vertices to $G$ connected by an edge).  
   We construct $G'$ as follows. Take two copies of $V$, call them $V_1,V_2$ and
   add all possible edges between vertices of $V_2$. For each $u\in V$, we connect
   $u\in V_1$ with all vertices $v\in V_2$ such that $v\in N_G[u]$, i.e.\ all vertices $v$ that are neighbors of $u$ in $G$. Let us also
   define the capacities: each $u\in V_1$ has capacity equal to its degree in $G$; each $u\in V_2$
   has capacity $\frac{n-k-1}{2}$.  This completes the construction. $G'$ is
   chordal because it is a split graph.
   
   Suppose that $G$ has a dominating set of size $k$. We delete the corresponding vertices
   of $V_2$ and claim that $G'$ becomes orientable. We observe that all vertices
   of $V_1$ have at least a deleted neighbor, since we deleted a dominating set of
   $G$, hence for each such vertex the number of remaining incident edges is at
   most its capacity.  We therefore orient all edges incident on $V_1$ towards
   $V_1$.  Finally, for the remaining vertices of $V_2$ which induce a clique of
   size $n-k$ we orient their edges using Lemma \ref{lem:clique} so that they all
   have in-degree exactly $\frac{n-k-1}{2}$.

   For the converse direction, suppose we can delete at most $k$ vertices of the
   new graph to make it orientable respecting the capacities. Again, as in Theorem
   \ref{thm:whard-chordal} we assume we have a solution of size exactly $k$,
   otherwise we add some vertices.  Furthermore, any used vertex of $V_1$ can be
   exchanged with one of its neighbors in $V_2$, since all vertices of $V_1$ have
   degree one more than their capacities, hence we assume that the solution
   deletes $k$ vertices of $V_2$.  We show that these vertices are a dominating
   set of $G$.  Suppose for contradiction that they are not, so $u\in V_1$ does
   not have any deleted neighbors in $V_2$.  Since the number of edges
   connecting $u\in V_1$ to $V_2$ is equal to the degree of $u$ plus 1, at least one of them is oriented towards $V_2$.
   But now the $n-k$ non-deleted vertices of $V_2$, because of Lemma
   \ref{lem:clique} all have in-degree exactly equal to the capacities inside the
   clique they induce. Hence, the additional edge from $V_1$ will force a vertex
   to violate its capacity.  \end{proof}

\begin{theorem}\label{thm:chordal-alg}
\dOr\ can be solved in time $d^{O(d+k)}\cdot n^{O(1)}$ on chordal graphs, where $k$
is the size of the solution.
\end{theorem} 
\begin{proof}
 
 This proof relies on standard techniques (dynamic programming on tree
 decompositions), so we will sketch some of the details. Recall that given a
 chordal graph $G$, it is known that we can obtain in polynomial time an optimal
 tree decomposition of $G$ of width $\omega(G)$, where $\omega(G)$ is the
 maximum clique size of $G$. We now observe that, because of Lemma
\ref{lem:clique}, we can assume that $\omega(G)< 2d+k+2$, because if the graph
 contains a clique on $2d+k+2$ vertices, even after deleting $k$ vertices we
 will not be able to produce a $d$-orientation and we can immediately reject. We
 therefore have $\tw(G)\le 2d+k+1$.
 
 Our algorithm now performs standard dynamic programming on the given tree
 decomposition, similarly to the algorithm of \cite{BodlaenderOO16}: we maintain
 a table in each bag which, for each vertex of the bag states either that the
 vertex has been deleted, or its in-degree in the orientation of the current
 partial solution. Since the in-degree is a number in $\{0,\ldots,d\}$, each
 vertex has $d+2$ possible states. This makes the total size of the DP table at
 most $(d+2)^{\tw} \le (d+2)^{2d+k+1}$. It is now not hard to see that such a
 table can be updated in time polynomial in its size, giving us a solution at
 the root of the tree decomposition.  \end{proof}

 
 \section{SETH Lower Bound for Treewidth}\label{sec:seth_tw}

\subparagraph*{Overview.} We follow the approach for proving SETH lower bounds
for treewidth algorithms introduced in \cite{LokshtanovMS18} (see also Chapter
14 in \cite{CyganFKLMPPS15}), that is, we present a reduction from \textsc{SAT}
to \dOr, for any fixed $d\ge1$, showing that if there exists a better than $(d+2)^{\tw}$ algorithm for
\dOr, we obtain a better than $2^n$ algorithm for \textsc{SAT}.

Similarly to these proofs, our reduction is based on the construction of ``long
paths'' of \emph{Block gadgets}, that are serially connected in a path-like
manner. Each such ``path'' corresponds to a group of variables of the given
formula, while each \emph{column} of this construction is associated with one
of its clauses. Intuitively, our aim is to embed the $2^n$ possible variable
assignments into the $(d+2)^\tw$ states of some optimal dynamic program that
would solve the problem on our constructed instance. The hard part of the
reduction is to take the natural $d+2$ options available for each vertex,
corresponding to its in-degree ($d+1$) or the choice to delete it ($+1$), and
use them to compress $n$ boolean variables into roughly $\frac{n}{\log(d+2)}$
units of treewidth.

Below, we present a sequence of gadgets used in our reduction. The
aforementioned block gadgets, which allow a solution to choose among $d+2$
reasonable choices, are the main ingredient. We connect these gadgets in a
path-like manner that ensures that choices remain consistent throughout the
construction, and connect clause gadgets in different ``columns'' of the
constructed grid in a way that allows us to verify if the choice made
represents a satisfying assignment, without increasing the graph's treewidth.

\subparagraph*{OR gadget.} We use an OR gadget with two endpoints $v,u$  whose purpose is to ensure that in any optimal solution, either $v$ or $u$ will have to be deleted. This gadget is simply a set of $2d+2$ vertices of capacity 1, connected to both $v$ and $u$, as shown in Figure \ref{fig:seth_orientable_tw_or}.

\begin{figure}[htbp]
\centering
  \includegraphics[width=40mm]{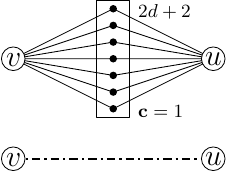}
\caption{An example OR gadget. In the following, OR gadgets are shown as dotted edges.}
 \label{fig:seth_orientable_tw_or}
\end{figure}

\begin{figure}[htbp]
  \centering
  \includegraphics[width=50mm]{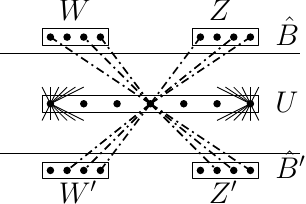}
  \caption{Example connections between a set $U$ and the $p$ sets $W,Z$ in the gadgets $\hat{B}$ of its group.}
  \label{fig:seth_orientable_tw_assignment}
\end{figure}

\begin{lemma}\label{or_gadget}
Let $G=(V,E), \mathbf{c}:V\to\{0,\ldots,d\}$ be an instance of \cdOr, and $u,v$
be two vertices connected via an OR gadget. Then, any optimal solution must
delete at least one of $u,v$.  \end{lemma} 
\begin{proof}

Suppose there exists an optimal solution that does not delete $u$ or $v$.
Because there are $2d+2$ internal vertices in the gadget, and at most $2d$ of
the edges of the gadget can be oriented towards $u$ or $v$, there are at least
two internal vertices of the gadget that are deleted by the solution. We place
these vertices back in the graph and delete $u$ in their place. We now have a
strictly smaller solution, for which we can obtain a valid orientation, since
the re-introduced vertices have capacity $1$ and degree $1$. This contradicts
the optimality of the initial solution.\end{proof}

\begin{lemma}\label{or_gadget2}

Let $G=(V,E), \mathbf{c}:V\to\{0,\ldots,d\}$ be an instance of \cdOr, and $G'$
be the graph obtained from $G$ by deleting, for each pair of vertices $u,v\in V$ which
are connected through an OR gadget, all internal vertices of the
gadget and adding the edge $(u,v)$ if it does not exist. It is $\pw(G)\le
\pw(G')+1$.

\end{lemma}

\begin{proof}

Consider a path decomposition of $G'$ and an OR gadget of $G$ between vertices
$u,v$. Since $G'$ contains the edge $(u,v)$, there is a bag containing both $u$
and $v$. We insert immediately after this bag $2d+2$ copies of it, and insert
into each copy one of the internal vertices of the OR gadget. Repeating this
for all OR gadgets gives a path decomposition of $G$ of width $1$ more than the
original decomposition.  \end{proof}

\subparagraph*{Clause gadget $\hat{C}(N)$.}
This gadget is responsible for determining clause satisfaction, based on the choices made in the rest of the graph. It is identical to the one used for \textsc{Independent Set} in \cite{LokshtanovMS18}, as finding a maximum independent set can be seen as equivalent to finding a minimum-sized deletion set for 0-orientability.
The gadget consists of a sequence of $N$ triangles, where every pair of consecutive triangles is connected by two edges, along with two pendant vertices attached to the first and last triangle. We set the capacities of all vertices to 0 and refer to the $N$ degree-2 vertices within the triangles as the \emph{inputs}. Figure \ref{fig:seth_orientable_tw_clause} provides an illustration.
For a clause $C_{\mu}$ with $q_{\mu}$ literals, the gadget's purpose is to offer an $1$-in-$q_{\mu}$ choice, while its pathwidth remains constant: in any valid solution, there must be at least two vertices deleted from every triangle, but not all of the inputs, as this would leave a non-orientable edge between two capacity-0 vertices. In our final construction, the non-deleted leaf will correspond to some true literal within the clause.

\begin{figure}[htbp]
\centerline{\includegraphics[width=70mm]{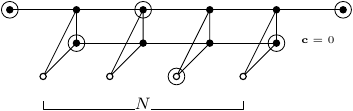}}
\caption{An example $\hat{C}(N)$. Note circled vertices forming an independent set, the rest being a minimum deletion set for 0-orientability.}
\label{fig:seth_orientable_tw_clause}
\end{figure}

 \begin{lemma}\label{clause_gadget}
  The minimum size of a deletion set $K$ in $\hat{C}(N)$ is equal to $2N$, while this cannot include all input vertices and $\pw(\hat{C}(N))\le3$.
 \end{lemma}
\begin{proof}
 We observe that, due to the capacities of all vertices being set to 0, there can be no edge left in the graph after deletion of $K$, meaning $\hat{C}(N)\setminus K$ must be an independent set. The lemma then follows from lemmas 2 and 3 of \cite{LokshtanovMS18} (see also Claim~14.39 in \cite{CyganFKLMPPS15}).
\end{proof}

\subparagraph*{Block gadget $\hat{B}$.} This gadget is the basic building block of our construction:
\begin{enumerate}
 \item Make three vertices $a,a',b$. Note that in the final construction, our block gadgets will be connected serially, with vertex $a'$ being identified with the following gadget's vertex $a$.
 \item Make three independent sets $X\coloneqq\{x_1,\dots,x_d\}$, $Y\coloneqq\{y_1,\dots,y_d\}$, $Q\coloneqq\{q_1,\dots,q_{2d+1}\}$.
 \item Make two sets $W\coloneqq\{w_0,\dots,w_{d+1}\}$ and $Z\coloneqq\{z_0,\dots,z_{d+1}\}$. 
 \item Connect all vertices of $X$ with vertex $a$ and with all vertices of $Q$.
 \item Connect all vertices of $Y$ with vertex $a'$ and with all vertices of $Q$.
 \item Connect all vertices from $W$ except $w_{d+1}$ to $b$ and all vertices of $X$.
 \item Connect all vertices from $Z$ except $z_{d+1}$ to $b$ and all vertices of $Y$.
 \item Attach OR gadgets between the pairs: $a$ and $b$, $b$ and $a'$, $a$ and $w_{d+1}$, $a'$ and $z_{d+1}$.
 \item Attach OR gadgets between any pair of vertices in $W\cup Z$, except for the pairs $(w_i,z_i)$ for $i\in\{0,\ldots,d+1\}$. In other words, $W\cup Z$ is an OR-clique, minus a perfect matching.
\end{enumerate}

We set the capacities as follows (see also Figure \ref{fig:seth_orientable_tw_block}).

\begin{itemize}
 \item $\mathbf{c}(a)=\mathbf{c}(a')=d$, $\mathbf{c}(b)=0$.
 \item $\forall i\in[1,2d+1],\mathbf{c}(q_i)=d$, and $\forall i,j\in[1,d],\mathbf{c}(x_i)=\mathbf{c}(y_j)=0$.
 \item $\forall i\in[0,d],\mathbf{c}(w_i)=i,\mathbf{c}(z_i)=d-i$, and $\mathbf{c}(w_{d+1})=\mathbf{c}(z_{d+1})=0$.
\end{itemize}

Intuitively, there are $d+2$ options in each gadget, linked to the
circumstances of vertices $a,a'$ and $b$:\footnote{Each such option can be seen
to correspond with one of the states that some optimal dynamic programming
algorithm for the problem would assign to vertex $a$: it is either deleted, or
has a number $i\in[0,d]$ of incoming edges within the gadget.} there will have
to be $d$ vertices deleted in total from $X\cup Y$ and the numbers will be
complementary: if $i$ vertices remain in $X$ then, due to $Q$ being of size
$2d+1$ (it is never useful to
delete any of them), there must be $d-i$ vertices remaining in $Y$. Thus, the
$d+1$ options can be seen as represented by the number of vertices remaining in
$X$, while for each one, vertex $b$ must also be deleted due to the OR gadgets
connecting it to $a,a'$.  The extra option is to ignore the actual number of
deletions within $X$ and remove both $a,a'$ instead.

The sets $W,Z$ are connected in such a way that any reasonable feasible
solution will delete all of their vertices, except for a pair $w_i,z_i$ for
some $i\in\{0,\ldots,d+1\}$. The non-deleted pair is meant to encode a choice
for this block gadget.

\begin{figure}[htbp]
\centerline{\includegraphics[width=130mm]{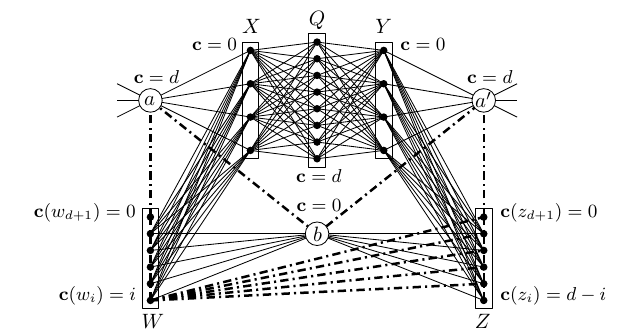}}
\caption{Our block gadget $\hat{B}$. Capacities are shown next to vertices/sets, the OR-connections within $W,Z$ are shown as paths, while the OR-connections between $W,Z$ are only shown for $w_0$.}
\label{fig:seth_orientable_tw_block}
\end{figure}

\subparagraph*{Global construction.} Fix some integer $d$, and suppose that for
some $\epsilon>0$ there exists a $O^*((d+2-\epsilon)^{\tw})$ algorithm for \dOr. We
give a reduction which, starting from any \textsc{SAT} instance with $n$
variables and $m$ clauses, produces an instance of \dOr, such that applying
this supposed algorithm on the new instance would give a better than $2^n$
algorithm for \textsc{SAT}.

We are faced with the problem that $d+2$ is not a power of $2$, hence we will need to create a correspondence between groups of variables of the \textsc{SAT} instance and groups of block gadgets.
We first choose an integer
$p=\lceil\frac{1}{(1-\lambda)\log_2(d+2)}\rceil$, for $\lambda=\log_{d+2}(d+2-\epsilon)<1$. We then group
the variables of $\phi$ into $t=\lceil\frac{n}{\gamma}\rceil$ groups
$F_1,\dots,F_t$, where $\gamma=\lfloor\log_2(d+2)^p\rfloor$ is the maximum
size of each group. Our construction then proceeds as follows (see Figure \ref{fig:seth_orientable_tw_global}):
\begin{enumerate}
 \item Make a group of $p$ block gadgets $\hat{B}_{\tau}^{1,\pi}$ for $\pi\in[1,p]$, for each group $F_{\tau}$ of variables of $\phi$ with $\tau\in[1,t]$.
 \item Make a clique $U_{\tau}^1\coloneqq\{u_{\tau}^{1,1},\dots,u_{\tau}^{1,(d+2)^p}\}$ on $(d+2)^p$ vertices, whose capacities are all set to 0, for each group $F_{\tau}$ of variables of $\phi$ with $\tau\in[1,t]$.
 \item For each $\tau\in[1,t]$, associate each of these $(d+2)^p$ vertices from $U_{\tau}^1$ with one of the $d+2$ options for deletion \emph{for each} gadget $\hat{B}_{\tau}^{1,\pi}$, i.e.\ each vertex of $U_{\tau}^1$ corresponds to a set containing $d+2$ vertices from each gadget $\hat{B}_{\tau}^{1,\pi}$ for $\pi\in[1,p]$.
 \item Connect each $u_{\tau}^{1,i}$ for $i\in[1,(d+2)^p]$ to each vertex from each $W$ and $Z$ within each of the $p$ gadgets $\hat{B}$ that \emph{do not match} the option associated with $u_{\tau}^{1,i}$ via OR gadgets (see Figure \ref{fig:seth_orientable_tw_assignment} for an example).
 \item Make $m(tpd+2)$ copies of this first ``column'' of gadgets.
 \item Identify each vertex $a'$ in $\hat{B}_{\tau}^{l,\pi}$ with the vertex $a$ of its following gadget $\hat{B}_{\tau}^{l+1,\pi}$, i.e.\ for fixed $\tau\in[1,t]$ and $\pi\in[1,p]$, all block gadgets are connected in a path-like manner.
 \item For every clause $C_{\mu}$, with $\mu\in[1,m]$, make a clause gadget $\hat{C}_{\mu}^h$ with $N=q_{\mu}$ inputs,
 where $q_{\mu}$ is the number of literals\footnote{We assume that $q_{\mu}$ is always even, by duplicating some literals if necessary.} in clause $C_{\mu}$ and for each $h\in[0,tpd+1]$.
 \item For every $\tau\in[1,t]$, associate one of the $(d+2)^p$ vertices of $U_{\tau}^l$ (that is in turn associated with one of $d+2$ options for each of the $p$ block gadgets of group $F_{\tau}$), with an assignment to the variables in group $F_{\tau}$. In this way, an assignment to the variables in group $F_{\tau}$ corresponds to a set of $(d+2)^p$ vertices from $p$ block gadgets $\hat{B}_{\tau}^{l,\pi}$ ($d+2$ from each one), as well as a vertex of $U_{\tau}^l$. Note that as there are at most $2^{\gamma}=2^{\lfloor\log_2(d+2)^p\rfloor}$ assignments to the variables in $F_{\tau}$ and $(d+2)^p\ge2^{\gamma}$ such vertices, the association can be unique for each $\tau$ (and the same for all $l\in[1,m(tpd+2)]$).
 \item Each of the clause gadget's $q_{\mu}$ inputs will correspond to a literal appearing in clause $C_{\mu}$.
 \item Connect via OR gadgets each input from each $\hat{C}_{\mu}^h$, corresponding to a literal whose variable appears in group $F_{\tau}$, to the all vertices from the set $U_{\tau}^{mh+\mu}$ (in its appropriate column) whose associated assignments \emph{do not satisfy} the input's literal.
\end{enumerate}

\begin{figure}[htbp]
\centerline{\includegraphics[width=140mm]{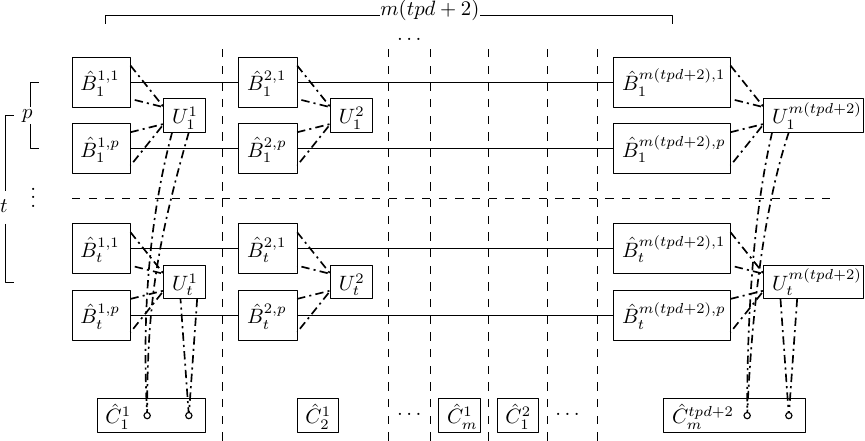}}
\caption{A simplified picture of the complete construction.}
\label{fig:seth_orientable_tw_global}
\end{figure}

Intuitively, our block gadgets $\hat{B}$ connected on long paths are responsible for maintaining a consistent selection of vertices over each line that corresponds to an assignment to the variables of each group. These selections are then reflected on the situation of sets $U$, that are in turn responsible for passing these selections in each column to the clause gadgets $\hat{C}$ to ensure each clause is satisfied. Due to the possible shift in selections/assignments over each line, this arrangement is repeated $tpd+2$ times to ensure that at least one of the repetitions of consecutive columns (one for each clause) will contain a consistent assignment whose satisfiability we can check.

\begin{lemma}\label{lem_seth_tw_fwd}
 If $\phi$ has a satisfying assignment, then there exists $K\subseteq V(G)$ such that $G\setminus K$ is $d$-orientable, with $k=|K|=(tpd+2)\sum_{\mu=1}^{m}(2q_{\mu}+t(3p(d+1)+(d+2)^p-1))$.
\end{lemma}
\begin{proof}
 Given a satisfying assignment for $\phi$ we will show the existence of a deletion set $K$ in $G$ of size $k$, whose removal from $G$ would leave the graph $d$-orientable. We first identify, for each $\tau\in[1,t]$, the vertex in $U_{\tau}^l$ (being the ``same'' for all $l\in[1,m(tpd+2)]$) associated with the partial assignment to the group $F_{\tau}$ of variables and include in $K$ all other vertices of $U_{\tau}^l$, for each $\tau,l$. Further, in each block gadget $\hat{B}$, we identify the corresponding option for deletion and include in $K$ all such vertices: vertex $a$ or $b$ and $d$ vertices in total from $X\cup Y$ (that correspond to the vertex deleted from $U$, for all $p$ block gadgets this $U$ is connected to), as well as the $2(d+1)$ vertices from sets $W,Z$ that do not match this option. Finally, we include in $K$ two vertices per triangle from each clause gadget $\hat{C}$ that do not match the above selections from the sets $U$: as
 this is a satisfying assignment, there will be at least one literal that is satisfied in each clause and the set $U$ of the group $F_{\tau}$ that the literal's variable appears in will contain only one vertex that is not already included in $K$, while the clause gadget will have one input vertex that is connected to the \emph{other} vertices of this $U$ (via OR gadgets), i.e.\ the vertices whose associated assignments do not satisfy the literal. We thus include all other inputs in $K$, along with the other vertex per triangle whose removal would leave an independent set within $\hat{C}$ (by Claim~14.39 in \cite{CyganFKLMPPS15}, we can always find such an independent set).
 
 Concerning the size of $K$, we have $3(d+1)$ vertices from each block gadget $\hat{B}$, $t(3p(d+1)+(d+2)^p-1)$ vertices from all $t$ $p$-groups of block gadgets, $2q_{\mu}+t(3p(d+1)+(d+2)^p-1)$ for each column, giving for each $m$-section of columns $\sum_{\mu=1}^{m}(2q_{\mu}+t(3p(d+1)+(d+2)^p-1))$ (for one cycle through all clauses), that finally gives for all $tpd+2$ $m$-sections a total of $k=(tpd+2)\sum_{\mu=1}^{m}(2q_{\mu}+t(3p(d+1)+(d+2)^p-1))$. What remains is to show that $G\setminus K$ is indeed $d$-orientable.
 
 In each block gadget $\hat{B}$, if $b\in K$, we orient all remaining edges from $X$ to $a$ and from $Y$ to $a'$ (the following gadget's $a$). As the capacity of $a$ vertices is $d$ and there are $i$ vertices deleted from each $X$ and $d-i$ from each $Y$ (on the same line), this orientation does not exceed the limit of $d$ for any vertex $a$. Next, we orient all remaining edges from $X$ and $Y$ towards the vertices of $Q$, whose capacities of $d$ are again sufficient. Since the vertices remaining in $W,Z$ (not $w_{d+1},z_{d+1}$, as $b\in K$) match the number of deletions within $X,Y$, we can safely orient all edges from $X,Y$ towards the two remaining vertices in $W,Z$, whose capacities will be exactly sufficient (observe that $b\in K$ implies all edges from $b$ to $W,Z$ are also deleted). Further, all OR gadgets are satisfied, as at least one of their endpoints are in $K$: we only retain one vertex from each $W$ and $Z$ (that are not connected via an OR gadget as they match the \emph{right} deletions from $X,Y$), while vertices $b,w_{d+1},z_{d+1}\in K$. On the other hand, if $a\in K$, then $a'\in K$ as well and the vertices remaining in $W,Z$ are $w_{d+1},z_{d+1}$, whose OR gadgets to $a,a'$ are satisfied, while there is no edge between them and $b$. The division of deletions within $X,Y$ is now irrelevant but their number must still add up to $d$, thus we can again orient all edges from $X,Y$ towards $Q$, while there are no remaining edges between $X,Y$ and $W,Z$.
 
 Concerning the connections between block and clause gadgets, the OR gadgets are again satisfied as $K$ includes at least one of their endpoints and we can safely orient all edges from the other endpoint towards the gadget's inner vertices: there is only one vertex remaining from each $W,Z$ and these complement the only vertex remaining from their connected set $U$, that in turn complements the remaining input from the column's clause gadget $\hat{C}$.
 Within each $\hat{C}$, the remaining 0-capacity vertices form an independent set.
 For the capacity cliques attached to each vertex, we orient all internal edges such that the vertices of the clique all have $d$ incoming edges each and the edges towards the vertex the clique is attached to are oriented towards it, thus fixing its capacity, that is externally satisfied as described above (see also Lemma \ref{lem:saturation}).
\end{proof}

\begin{lemma}\label{lem_seth_tw_bwd}
 If there exists $K\subseteq V(G)$ such that $G\setminus K$ is $d$-orientable, with $k=|K|=(tpd+2)\sum_{\mu=1}^{m}(2q_{\mu}+t(3p(d+1)+(d+2)^p-1))$, then $\phi$ has a satisfying assignment.
\end{lemma}
\begin{proof}
 Given a deletion set $K$ of size $k$ such that $G\setminus K$ is $d$-orientable, we will show the existence of a satisfying assignment for $\phi$. To identify the necessary structure of $K$, we first observe that the number of $m$-sections of columns (corresponding to a complete cycle over all clauses) being $tpd+2$, it suffices to show that for each $m$-section, any deletion set $K$ will have to include $\sum_{\mu=1}^{m}(2q_{\mu}+t(3p(d+1)+(d+2)^p-1))$ vertices, meaning that in each column the number of vertices included in $K$ must be $2q_{\mu}+t(3p(d+1)+(d+2)^p-1)$, by showing that $G\setminus K$ would not be $d$-orientable otherwise.
 
 By Lemma \ref{clause_gadget}, there will have to be at least $2q_{\mu}$ vertices included in $K$ from every $\hat{C}$, while this cannot include all its inputs. This leaves $t(3p(d+1)+(d+2)^p-1)$ vertices for all $t$ groups and we again aim to show that each group will require the deletion of at least $3p(d+1)+(d+2)^p-1$ vertices.
 
 As each set $U$ is a clique on $(d+2)^p$ vertices of capacity 0, at least all but one of them will have to be included in $K$, leaving $3p(d+1)$ vertices for the $p$ block gadgets of each group, with $3(d+1)$ required from each $\hat{B}$: sets $W,Z$ are OR-cliques, meaning that at least all but one vertex can be retained from each, thus accounting for the $2(d+1)$ vertices, while each vertex in set $Q$ is of capacity $d$ with $2d$ edges from $X,Y$. As $|Q|=2d+1$, set $K$ cannot include all of them and deleting fewer vertices from $Q$ would not suffice to orient the remaining edges. Thus we know there must be at least $d$ vertices from sets $X,Y$ included in $K$, while for the final vertex from $\hat{B}$ we note that as $a$ and $b$ are connected via an OR gadget, at least one of them will also have to be in $K$. This gives the correct number of vertices included in $K$ from each part of $G$.
 
 If in each $\hat{C}$ there is one input that is not included in $K$, then all vertices in the sets $U$ of this column that it is connected to (via OR gadgets) must be in $K$. This means the single vertices remaining in sets $U$ in each column must match the inputs remaining in gadgets $\hat{C}$, as every input is connected to all vertices associated with some partial assignment that \emph{does not} satisfy its corresponding literal. Also, each vertex from each set $U$ is connected to all but one vertex of each of the $p$ pairs of sets $W,Z$ in the group of block gadgets $\hat{B}$ of this set $U$ (again via OR gadgets). Thus the vertices remaining within each $W,Z$ must match the vertex retained from their connected set $U$. Then in each $\hat{B}$, due to the OR gadgets between non-matching vertices of $W,Z$, we know the vertices retained from $W$ and $Z$ must agree, i.e.\ they must be of the same index $i$. For $i\not=d+1$, this implies $b\in K$ (due to the edges from $b$ to both $w_{d+1},z_{d+1}$) and further, that the number of vertices in $K$ from $X$ is $i$, while the number of vertices in $K$ from $Y$ is $d-i$ (due to the capacities of $w_{d+1},z_{d+1}$). For $i=d+1$, we must have $a\in K$ and also $a'\in K$. Thus the number of deletions within sets $X$ and $Y$ of each block gadget $\hat{B}$, along with deletion of either $a$ or $b$ (i.e.\ the options within $\hat{B}$) must match the vertices remaining in sets $W,Z$, that in turn must match (complement) the vertex remaining in set $U$ (over all $p$ block gadgets of that group), that must finally also match (complement) the input remaining in that column's clause gadget $\hat{C}$.
 
 Next, we require that there exists at least one $h\in[0,tpd+1]$ for every $\tau\in[1,t]$ for which the selected options within all $p$ block gadgets $\hat{B}_{\tau}^{mh+\mu,\pi}$ (and thus also the corresponding retained vertices from sets $U$) do not change for all $\mu\in[1,m]$, i.e.\ that there exists some $m$-section ($m$ successive columns) for which the selected options remain unchanged (over all rows). Observe that, if $i$ vertices are included in $K$ from a set $X$ and $d-i$ from $Y$ in some gadget $\hat{B}_{\tau}^{l,\pi}$ and $a,a'\notin K$, then $a'$ must have $i$ incoming edges from $Y$, which implies it can have at most $d-i$ incoming edges from the set $X$ of the following gadget $\hat{B}_{\tau}^{l+1,\pi}$ (inside which it is the vertex $a$). If, on the other hand, $b\notin K$ in some gadget $\hat{B}_{\tau}^{l,\pi}$, then both $a,a'$ must be in $K$ due to the OR gadgets between $b$ and $a,a'$ and the budget ($3(d+1)$) for the following gadget $\hat{B}_{\tau}^{l+1,\pi}$ implies that $b\notin K$ there as well. This means the options for deletion over a row of gadgets (fixed $\tau\in[1,t],\pi\in[1,p]$) are \emph{monotone}: if the selected option in some gadget $\hat{B}_{\tau}^{l,\pi}$ corresponds to deleting $i$ vertices from $X$, then the selected option in the following gadget $\hat{B}_{\tau}^{l+1,\pi}$ must correspond to deleting $i'$ vertices from $X$, with $i'\ge i$, while if vertex $b$ is retained in some gadget (even if it was deleted in its predecessor) then it must also be retained in its follower. Thus this ``shift'' in options can happen at most $d+1$ times for each $\pi\in[1,p]$, giving $p(d+1)$ times for each $\tau\in[1,t]$, or $tp(d+1)$ times over all $\tau$. By the pigeonhole principle, there must thus exist an $h\in[0,tp(d+1)]$ such that no shift happens among the gadgets $\hat{B}_{\tau}^{mh+\mu,\pi},\forall\tau\in[1,t],\pi\in[1,p],\mu\in[1,m]$.
 
 Our assignment for $\phi$ is then given by the selected options within each gadget $\hat{B}_{\tau}^{mh+1}$ for this $h$: for every group $F_{\tau}$ we consider the single remaining vertex in $U_{\tau}^{mh+1}$ that is associated with a partial assignment for the variables in $F_{\tau}$. As there must be exactly one remaining vertex in each set $U$ that matches both the selected options within all $p$ gadgets of its group, as well as the remaining input within that column's clause gadget $\hat{C}$, these selected options being unchanged over the complete $m$-section, this assignment also satisfies each clause $C_{\mu}$ for all $\mu\in[1,m]$. 
\end{proof}

\begin{lemma}\label{lem_seth_tw_bound}
 Graph $G$ has treewidth $\tw(G)\le tp+f(d,\epsilon)$, for $f(d,\epsilon)=O(d^p)$.
\end{lemma}
\begin{proof}
 We will show a pathwidth bound of $\pw(G)\le tp+O(d^p)$ by providing a mixed search strategy to clean $G$ using at most this many searchers simultaneously. The treewidth bound then follows from Lemma \ref{lem:widths} and the well known relationship between pathwidth and \emph{mixed search number} $\ms$ (see \cite{TAKAHASHI1995253}): \emph{for any graph $G$, it is $\pw(G)\le\ms(G)\le\pw(G)+1$}.
 
 We initially place one searcher on every vertex $a$ of gadgets $\hat{B}_{\tau}^{1,\pi}$ for all $\tau\in[1,t],\pi\in[1,p]$. These account for the $tp$ searchers that we will be moving to the vertices $a$ of the following gadgets on each line after the inner parts of each column have been cleaned, while the remaining searchers will be (re)used to clean each part of each column.
 
 In particular, we assume the input vertices of $\hat{C}_1^1$ are ordered in terms of the group $F_{\tau}$ of variables that their associated literals appear in and place a searcher on each vertex of the set $U_{\tau}^1$ that the first input vertex of $\hat{C}_1^1$ is connected to, along with three searchers on the vertices of the first triangle. Using some of the remaining searchers we first clean all paths between this input and $U_{\tau}^1$. Note that OR gadgets involve $2d+3$ vertices, with only 3 searchers being sufficient (Lemma \ref{or_gadget2}). Concerning the cliques attached to each vertex to set its capacity, see Lemma \ref{lem:saturation}. If the next input is also associated with a variable in group $F_{\tau}$, we slide the searchers on the next triangle. If the next input is associated with a variable appearing in some following group $F_{\tau'}$, we keep the searchers on the same triangle until this block gadget has been cleaned and we reach the group of variables (and block gadgets) $F_{\tau'}$.
  
 To clean a block gadget $\hat{B}_{\tau}^{1,\pi}$, we place a searcher on each vertex of sets $X,Y,W,Z,Q$. These are $<12d$ and we can completely clean the gadget (and all vertices between it and $U_{\tau}^1$) using the remaining searchers. We then remove the searcher from vertex $a$ and place it on $a'$ (to clean the following gadget after the first column has been cleaned). Having thus cleaned $\hat{B}_{\tau}^{1,\pi}$ we remove all searchers from it and place them on the corresponding vertices of $\hat{B}_{\tau}^{1,\pi'}$ and repeat the process. After all $p$ gadgets for $\tau$ have been cleaned in this way, the initial searchers having moved from vertices $a$ to $a'$, we can remove all searchers from $U_{\tau}^1$ and place them on $U_{\tau'}^1$ and repeat the process for $\tau'$. Performing the above for all $\tau\in[1,t]$ cleans the first column, along with the first clause gadget and we can then remove the three searchers from $\hat{C}_1^1$ and place them on the first triangle of $\hat{C}_2^1$ (see also Lemma \ref{clause_gadget}).
 
 Repeating the above process for all $m(tpd+2)$ columns completely cleans the graph and the number of searchers we have used simultaneously is $\le tp+O(d^p)$: our $tp$ initial searchers that we move from vertices $a$ to $a'$, three searchers for each clause gadget, $(d+2)^p$ searchers for each set $U$, at most $12d$ searchers for each block gadget $\hat{B}$ and at most $4d$ searchers for the smaller gadgets and remaining intermediate paths/edges.
\end{proof}

\begin{theorem}\label{thm_seth_tw}
 For any fixed $d\ge1$, if \textsc{$d$-Orientable Deletion} can be solved in $O^*((d+2-\epsilon)^{\tw(G)})$ time for some $\epsilon>0$, then there exists some $\delta>0$, such that \textsc{SAT} can be solved in $O^*((2-\delta)^n)$ time.
\end{theorem}
\begin{proof}
 Assuming the existence of some algorithm of running time $O^*((d+2-\epsilon)^{\tw(G)})=O^*((d+2)^{\lambda\tw(G)})$ for \textsc{$d$-Orientable Deletion}, where $\lambda=\log_{d+2}(d+2-\epsilon)$, and given a formula $\phi$ of \textsc{SAT}, we construct an instance of \textsc{$d$-Orientable Deletion} using the above construction and then solve the problem using the $O^*((d+2-\epsilon)^{\tw(G)})$-time algorithm. Correctness is given by Lemma \ref{lem_seth_tw_fwd} and Lemma \ref{lem_seth_tw_bwd}, while Lemma \ref{lem_seth_tw_bound} gives the upper bound on the running time:
 
 \begin{align}
  O^*((d+2)^{\lambda\tw(G)})&\le O^*\left((d+2)^{\lambda(tp+f(d,\epsilon))}\right)\\
  &\le O^*\left((d+2)^{\lambda p\left\lceil\dfrac{n}{\lfloor\log_2(d+2)^p\rfloor}\right\rceil}\right)\label{tw_comp_2}\\
  &\le O^*\left((d+2)^{\lambda p\dfrac{n}{\lfloor\log_2(d+2)^p\rfloor}+\lambda p}\right)\\
  &\le O^*\left((d+2)^{\lambda\dfrac{np}{\lfloor p\log_{2}(d+2)\rfloor}}\right)\label{tw_comp_4}\\
  &\le O^*\left((d+2)^{\delta'\dfrac{n}{\log_2(d+2)}}\right)\label{tw_comp_5}\\
  &\le O^*(2^{\delta''n})=O^*((2-\delta)^n)
 \end{align}
for some $\delta,\delta',\delta''<1$. Observe that in line (\ref{tw_comp_2}) the function $f(d,\epsilon)$ is considered constant, as is $\lambda p$ in line (\ref{tw_comp_4}), while in line (\ref{tw_comp_5}) we used the fact that there always exists a $\delta'<1$ such that $\lambda\dfrac{p}{\lfloor p\log_2(d+2)\rfloor}=\dfrac{\delta'}{\log_2(d+2)}$, as we have:
\begin{equation*}
 \begin{split}
    p\log_2(d+2)-1&<\lfloor p\log_2(d+2)\rfloor\\
    \Leftrightarrow\dfrac{\lambda p\log_2(d+2)}{p\log_2(d+2)-1}&>\dfrac{\lambda p\log_2(d+2)}{\lfloor p\log_2(d+2)\rfloor},\\
    \text{from which, by substitution, we get } \dfrac{\lambda p\log_2(d+2)}{p\log_2(d+2)-1}&>\delta',\\
    \text{now requiring } \dfrac{\lambda p\log_2(d+2)}{p\log_2(d+2)-1}&\le1,\\
    \text{or } p\ge\dfrac{1}{(1-\lambda)\log_2(d+2)},
 \end{split}
\end{equation*}
that is precisely our definition of $p$. This concludes the proof.
\end{proof}

\begin{corollary}\label{cor_seth_tw}
 If \textsc{Pseudoforest Deletion} can be solved in $O^*((3-\epsilon)^{\tw(G)})$ time for some $\epsilon>0$, then there exists some $\delta>0$, such that \textsc{SAT} can be solved in $O^*((2-\delta)^n)$ time.
\end{corollary}


\section{Algorithm for Clique-Width}\label{sec:cwalg}

In this section we present a dynamic programming algorithm for \dOr\ parameterized by the clique-width of the input graph, of running time $O^*(d^{O(d\cdot\cw)})$. The algorithm is based on the dynamic programming of \cite{Bodlaender2017} for \textsc{Max $W$-Light}, the problem of assigning a direction to each edge of an undirected graph so that the number of vertices of out-degree at most $W$ is maximized. As noted in \cite{Bodlaender2017} (and our Section \ref{sec:intro}), that problem is supplementary to \textsc{Min ($W+1$)-Heavy}, the problem of minimizing the number of vertices of out-degree at least $W+1$ in terms of exact computation (though their approximability properties may vary), that in turn can be seen as the optimization version of \dOr\ for $d=W$, if we simply consider the in-degree of every vertex instead of the out-degree (by reversing the direction of every edge in any given orientation).

The dynamic programming algorithm of \cite{Bodlaender2017} runs in XP-time $n^{O(d\cdot\cw)}$, by considering the full number of possible states for each label of a clique-width expression $T$ for the input graph $G$:\footnote{Slightly paraphrased here for \dOr, keeping the same notation.} for each node $t$ of $T$, it computes an \emph{in-degree-signature} of $G_t$ (i.e.\ the graph constructed up to the operation of $t$), being a table $A_t=(A_t^{i,j}),\forall i\in[1,\cw],j\in[0,d]$, if there is an orientation $\Lambda_t$ (of every edge of $G_t$) such that for each label $i\in[1,\cw]$ and \emph{in-degree-class} $j\in[0,d]$, the entry $A_t^{i,j}$ is the number of vertices labelled $i$ with in-degree $j$ in $G_t$ under $\Lambda_t$, and also a \emph{deletion set} $K_t$, where $K_t\coloneqq\bigcup_{i\in[1,\cw]}K_t^i$ for each $i\in[1,\cw]$, where $K_t^i$ is the set of vertices labelled $i$ that are deleted from $G_t$. Based on this scheme, the updating process of the tables is straightforward for Leaf, Relabel and Union nodes, while for Join nodes, the computation of the degree signatures is based on a result by \cite{Asahiro2012}, stating that an orientation satisfying any given lower and upper in-degree bounds for each vertex can be computed in $O(m^{3/2}\log n)$ time, where $m$ is the number of edges to be oriented. We refer to \cite{Bodlaender2017} for details.

The aim of this section is to improve the running time of the above algorithm to $O^*(d^{O(d\cdot\cw)})$, that is FPT-time parameterized by $d$ and $\cw$, by showing that not all of the natural states utilized therein are in fact required. The main idea behind this improvement is based on the redundancy of \emph{exactly} keeping track of the size of an in-degree-class above a certain threshold (i.e.\ $d^4$), since the valid $d$-orientations of a biclique created after joining such an in-degree-class with some other label are greatly constrained, as any optimal solution will always orient all new edges towards the vertices of this ``large'' class in order to maintain $d$-orientability (and update in-degree-class sizes accordingly), while respecting the given deletion set and orientation of previously introduced edges.

We formally define the notion of a (partial) \emph{valid solution}: for a node $t$ of the clique-width expression $T$, deletion set $K_t$ and orientation $\Lambda_t$, a (partial) \emph{solution} $X_t\coloneqq\{K_t\subseteq V(G_t),\Lambda_t\}$ is \emph{valid}, if $|K_t|\le k$ and $\delta_t^-(v)\le d,\forall v\in V(G_t)$ under orientation $\Lambda_t$. For two nodes $t_1,t_2$ of some clique-width expression $T$ where $t_2$ is a successor of $t_1$, we say that a solution $Y_{t_2}\coloneqq\{K_{t_2},\Lambda_{t_2}\}$ \emph{extends} a solution $X_{t_1}\coloneqq\{K_{t_1},\Lambda_{t_1}\}$, if $K_{t_1}\subseteq K_{t_2}$ and $\Lambda_{t_2}$ assigns the same directions as $\Lambda_{t_1}$ to all edges in $G_{t_1}$.
Further, we formally define a \emph{large} in-degree-class: for node $t$ and label $i\in[1,\cw]$, an in-degree-class $j\in[0,d]$ is \emph{large}, if its size is larger than $d^4$, or $A_t^{i,j}=|v\in V(G_t^i):\delta_t^-(v)=j|\ge d^4$.

\begin{lemma}\label{cw_algo_lem}
 For some node $t\in T$, label $i\in[1,\cw]$ and large in-degree-class $j\in[0,d]$ of $i$, if there exists a valid \emph{partial} solution $X_t=\{K_t,\Lambda_t\}$, that is extended to a valid \emph{global} solution $X_r=\{K_r,\Lambda_r\}$ for the root $r$ of $T$, where $\Lambda_t$ assigns a direction to some edges $(v,u)$ \emph{away from} vertices $v\in V_t^i$ belonging to the large in-degree-class $j$, then there exists another valid \emph{partial} solution $X'_t=\{K'_t,\Lambda'_t\}$ that is extended to a valid \emph{global} solution $X'_r=\{K'_r,\Lambda'_r\}$, with $|K'_r|\le |K_r|$, where $\Lambda'_t$ assigns direction \emph{towards} vertices $v$ belonging to the large in-degree-class $j$.
\end{lemma}
\begin{proof} Lemma \ref{biclique_nonorientable}
implies that no valid solution to \dOr\ can be obtained for a graph including
some biclique larger than $K_{2d,2d}$ (on any side) and in particular, that we
can assume that any Join node involves at most one label with a large
in-degree-class, while the other cannot have $2d$ vertices in total (all
other partial solutions can be discarded).
  
  Consider a Join node $t$ with $\eta_{p,q}$ for $p,q\in[1,\cw]$, where label
$p$ contains some large in-degree-class $j\in[0,d]$. As noted, we can assume
$|V_t^q|<2d$. Now, observe that the maximum number of newly created edges
between class $j$ and label $q$ that can be oriented towards the vertices of
$q$ is $<2d^2$, as each vertex of $q$ can have at most $d$ incoming edges for
the partial solution to remain valid. This means $>d^4-2d^2$ vertices of $j$
will have to be moved to another in-degree-class $j'>j$. The number of such shifts is bounded by the number $d$ of in-degree-classes for a partial solution to remain valid. Thus there will always be $>d^4-2d^3$ vertices in the
corresponding in-degree-class, which for $d\ge2$ is $>d^3$.
  
  Let $X_t=\{K_t,\Lambda_t\}$ be a valid partial solution that, during the computation of in-degree-signatures for some Join node $t$ with $\eta_{p,q}$, selects an orientation $\Lambda_t$ for the newly created edges that directs some edges $(v,u)$ away from a vertex $v$ of the large in-degree-class $j$ of label $p$ and towards a vertex $u$ of label $q$.
  Also let $Y_r=\{K_r,\Lambda_r\}$ be a valid global solution that extends $X_t$.
  Now consider global solution $Y'_r=\{K'_r,\Lambda'_r\}$ that differs from $Y_r$ by simply selecting an orientation $\Lambda'_r$ that differs from $\Lambda_r$ only in the direction of edges $(v,u)$:
  if $Y_r$ is valid, it is $|K_r|\le k$ and $\Lambda_r$ assigns no more than $d$ incoming edges to any vertex. As noted above, there must exist at least one vertex $w\in V_t^p$ belonging to large in-degree-class $j$, whose edges are all oriented towards it in $\Lambda_r$ and these must be $\le d$. Since $w$ and vertices $v$ (from large in-degree-class $j$ of label $p$, whose edges were oriented towards vertices $u$ of label $q$) belong to the same label $p$, any Join operations following node $t$ applied to $w$ are also applied to all $v$ and thus, the orientation $\Lambda'_t$ that directs $(u,v)$ towards $v$ can also be extended to an orientation $\Lambda'_r$ for global solution $Y'_r$ that does not require any more deleted vertices than $Y_r$ (keeping $|K'_r|\le|K_r|$) and remains valid as $\delta_r^-(w)=\delta_r^-(v)\le d$ for all such $v\in j$.
\end{proof}

\begin{theorem}\label{cw_algo_thm}
 Given a graph $G$ along with a $\cw$-expression $T$ of $G$, the \dOr\ problem can be solved in time $O^*(d^{O(d\cdot\cw)})$. 
\end{theorem}
\begin{proof}
 As explained above, our algorithm is a modification of the DP given in \cite{Bodlaender2017}, the necessary modifications being the following.
 First, due to our focusing on the decision version of \dOr, we consider that including a vertex in the deletion set and appropriately decreasing the given budget $k$ is decided upon while computing the table of its Leaf node. 
 Next, if for some $i\in[1,\cw],j\in[0,d]$ we have $A_t^{i,j}\ge d^4$, we replace the actual numerical value with a distinct marker and refer to this in-degree-class as large.
 
 We compute the tables of Union and Relabel nodes in the same way, noting that any in-degree-class that is combined through such a node's application with a large in-degree-class will have to become large in turn. Further, observe that if both labels involved in a Join operation contain at least one large in-degree-class then the implied partial solution cannot be valid, due to Lemma \ref{biclique_nonorientable}. This also implies that for a valid partial solution, if one label contains at least one large in-degree-class, then the size of the other label must be $<2d$.
 
 Then, by Lemma \ref{cw_algo_lem}, to update the tables of a Join node $t$ with $\eta_{p,q}$ for $p,q\in[1,\cw]$, where label $p$ contains some large in-degree-class $j\in[0,d]$, we assume that all newly created edges attached to the vertices of $j$ are oriented towards them in any orientation $\Lambda_t$ for this node, thus being able to correctly update the sizes of all in-degree-classes and both labels $p,q$: as all edges will be oriented towards the vertices of $j$, then all these vertices will have an increase of in-degree equal to the size of $q$, meaning that the in-degree-class $j+|V_t^q|$ of $p$ will now have to become large, while all other computations of orientations and in-degree-signatures remain unchanged. As we now require $(d^4)^{d+2}$ states for each label (instead of $n^{d+2}$), our running time is $O^*(d^{O(d\cdot\cw)})$.
\end{proof}


\section{W[1]-hardness for Clique-Width}\label{sec:cw_whard}

In this section we present a reduction establishing that the algorithm of
Section \ref{sec:cwalg} is essentially optimal. More precisely, we show that,
under the ETH, no algorithm can solve \dOr\ in time $n^{o(\cw)}$. As a result,
the parameter dependence of $d^{O(\cw)}$ of the algorithm in Section
\ref{sec:cwalg} cannot be improved to a function that only depends on $\cw$. We
prove this result through a reduction from \textsc{$k$-Multicolored Independent
Set}. As before, we employ capacities and implicitly utilize \cdOr.

\subparagraph*{Construction.} Recall that an instance $[G=(V,E),k]$ of
\textsc{$k$-Multicolored Independent Set} consists of a graph $G$ whose vertex
set is given to us partitioned into $k$ sets $V_1,\ldots V_k$, with $|V_i|=n$
for all $i\in [1,k]$, and with each $V_i$ inducing a clique.  Given such an
instance, we will construct an instance $G'=(V',E')$ of \textsc{$d$-Orientable
Deletion}, where $d=n$. Let $V_i\coloneqq\{v_i^1,\dots,v_i^n\},\forall
i\in[1,k]$. To simplify notation, we use $E$ to denote the set of
\emph{non-clique} edges, i.e.\ those connecting vertices in parts $V_i,V_j$ for
$i\not=j$. Our construction is given as follows, while Figure
\ref{fig:whard_cw_construction} provides an illustration: 

\begin{enumerate} 

\item Create two sets $A_i,B_i\subset V',\forall i\in[1,k]$ of $n$ vertices
each, of capacities 0.  

\item Make a set of \emph{guard} vertices $W_i,\forall i\in[1,k]$, of size
$kn+3|E|+1$, of capacities $n$.  

\item Connect each vertex of $W_i$ to all vertices of $A_i,B_i$ for all
$i\in[1,k]$.  

\item For each edge $e=(v_i^l,v_j^h)\in E$ with endpoints $v_i^l\in
V_i,v_j^h\in V_j$ (i.e.\ the $l$-th vertex of $V_i$ and the $h$-th vertex of
$V_j$), make four new vertices $a_l^e,b_l^e,a_h^e,b_h^e$.  

\item Connect $a_l^e,b_l^e,a_h^e,b_h^e$ to each other via OR gadgets.  

\item Connect $a_l^e$ to all vertices of $A_i$ and $b_l^e$ to all vertices of
$B_i$, while $a_h^e$ is connected to all vertices of $A_j$ and $b_h^e$ to all
vertices of $B_j$.  

\item Set the capacities $\mathbf{c}(a_l^e)=n-l-1$, $\mathbf{c}(b_l^e)=l-1$,
$\mathbf{c}(a_h^e)=n-h-1$ and $\mathbf{c}(b_h^e)=h-1$. 

\end{enumerate} 

\begin{figure}[htbp]
\centerline{\includegraphics[width=140mm]{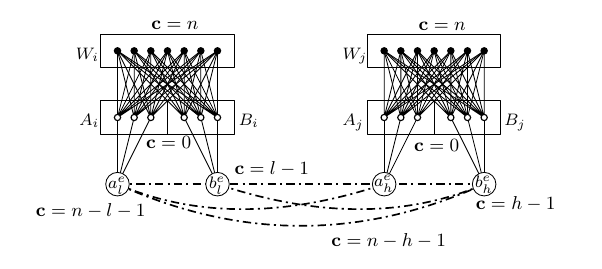}}
\caption{A partial view of the construction, depicting the gadgets encoding the
selection for $V_i,V_j$, as well the representation of an edge
$e=(v_i^l,v_j^h)$. Note dotted edges signifying OR gadgets.}
\label{fig:whard_cw_construction} \end{figure}

\begin{lemma}\label{lem_whard_cw_FWD}
 If $G$ has a $k$-multicolored independent set, then $G'$ has a deletion set $K'$, such that $G'\setminus K'$ is $d$-orientable, with $|K'|=kn+3|E|$.
\end{lemma}
\begin{proof}
 Let $K\subseteq V$ be a $k$-multicolored independent set in $G$ and $v_i^{l_i}$ denote the vertex selected from each $V_i$, or $K\coloneqq\{v_1^{l_1},\dots,v_i^{l_i},\dots,v_k^{l_k}\}$. Our deletion set $K'$ will include $l_i$ vertices from each $A_i$ and $n-l_i$ vertices from each $B_i$ (thus $n$ vertices in total will be deleted from each pair of $A_i,B_i$). Further, for each edge $e=(v_i^{l},v_j^{h})\in E$, our deletion set $K'$ will also include 3 out of the 4 vertices $a_{l}^e,b_{l}^e,a_{h}^e,b_{h}^e$ (it could not include less than 3, due to the OR gadgets anywhere between these four). We identify the vertices to include in $K'$ as follows: since $K$ is a $k$-multicolored independent set, we know there is no edge $e=(v_i^{l_i},v_j^{l_j})\in E$ and thus, at least one endpoint of the edge in question is different than the actual selection within $V_i$, i.e.\ either $l\not=l_i$, or $h\not=l_j$, or both (also for the other symmetrical cases). Without loss of generality we assume the selection within $V_i$ differs from the $v_i^l$ endpoint of $e$, or $l\not=l_i$. Thus, $K'$ includes $l_i$ vertices from $A_i$ and $n-l_i$ vertices from $B_i$. If $l>l_i$, then $l-1\ge l_i$ and the capacity of vertex $b_l^e$ will be sufficient for orientation of all (remaining) edges from $B_i$ towards it. We thus include in $K'$ all other three vertices $a_l^e,a_h^e,b_h^e$. On the other hand, if $l<l_i$, then $n-l-1\ge n-l_i$ and now the capacity of vertex $a_l^e$ will be sufficient for orientation of all (remaining) edges from $A_i$ towards it. We thus include in $K'$ all other three vertices $b_l^e,a_h^e,b_h^e$.
 
 In this way we have completed the deletion set $K'$ with $kn+3|E|$ vertices and what remains is to show that $G'\setminus K'$ is $d$-orientable. Having deleted $n$ vertices from each pair $A_i,B_i$, we can orient all edges from the remaining $n$ vertices in each $A_i,B_i$ towards each guard vertex in $W_i$, whose capacities are equal to $n$. Having also retained the correct vertex from each quadruple corresponding to some edge $e$ in $G$, we can safely orient all edges from either $A_i$ or $B_i$ (depending on whether the vertex that is not included in $K'$ is an $a_l^e$ or $b_l^e$) to the remaining vertex, whose capacity will be sufficient, as explained above. Finally, the edges remaining within the OR gadgets are oriented towards the vertices of the gadget, whose capacity of 1 is sufficient for the single remaining edge, as the other endpoint of the gadget is included in $K'$. 
\end{proof}

\begin{lemma}\label{lem_whard_cw_BWD}
 If $G'$ has a deletion set $K'$, such that $G'\setminus K'$ is $d$-orientable, with $|K'|=kn+3|E|$, then $G$ has a $k$-multicolored independent set.
\end{lemma}
\begin{proof} First, observe that due to the guard vertices $W_i$ being of
capacity $n$ and connected to $2n$ vertices each (whose capacities are equal to
0), at least $n$ vertices must be deleted from each pair of $A_i,B_i$ (included
in $K'$). Deleting any guard vertex does not decrease the number of edges that
must be oriented towards the remaining vertices of $W_i$ and the size of $K'$
does not allow for deletion of all guard vertices. Next, observe that due to the OR gadgets connecting all
four vertices $a_l^e,b_l^e,a_h^e,b_h^e$ corresponding to some edge $e=(v_i^l,v_j^h)\in
E$, at least three of them must be included in $K'$ for $G'\setminus K'$ to be
$d$-orientable. Since $|K'|=kn+3|E|$, then exactly $n$ vertices are deleted
from each pair of $A_i,B_i$ for each $i\in[1,k]$ and exactly 3 vertices are
deleted from each quadruple of vertices corresponding to some edge of $G$.
 
 Let $l_i\in[0,n]$ be the number of vertices deleted from each $A_i$ (meaning $n-l_i$ are deleted from $B_i$). We let our set $K$ include the $l_i$-th vertex from each $V_i$, i.e.\ $K\coloneqq\{v_1^{l_1},\dots,v_i^{l_i},\dots,v_k^{l_k}\}$ and claim $K$ is a $k$-multicolored independent set in $G$. Suppose that this is not the case, meaning there exists some edge $e=(v_i^{l_i},v_j^{l_j})\in E$. Consider the quadruple of vertices $a_{l_i}^e,b_{l_i}^e,a_{l_j}^e,b_{l_j}^e$ corresponding to edge $e$: as there are $l_i$ vertices deleted from $A_i$, there are $n-l_i$ (remaining) edges between $A_i$ and $a_{l_i}^e$. As the capacities of all vertices in $A_i$ are 0, all these edges must be oriented towards $a_{l_i}^e$, whose capacity is $n-l_i-1<n-l_i$. Further, there are $n-l_i$ vertices deleted from $B_i$ and thus $l_i$ edges need to be oriented towards $b_{l_i}^e$, whose capacity is $l_i-1<l_i$. This means both $a_{l_i}^e$ and $b_{l_i}^e$ must be included in $K'$ for $G'\setminus K'$ to be $d$-orientable. The same holds for $a_{l_j}^e,b_{l_j}^e$, however, whose capacities will also be 1 less than the required number of edges to be oriented towards them. This implies that if the vertices from each $V_i$ corresponding to the numbers of deletions within each pair of $A_i,B_i$ are not independent, then either set $K'$ is of size larger than $kn+3|E|$, or that $G'\setminus K'$ is not $d$-orientable, which is a contradiction.
\end{proof}

\begin{lemma}\label{lem_whard_cw_bound}
 The clique-width of $G'$ is $\cw(G')\le2k+O(1)$.
\end{lemma}
\begin{proof}
 We will describe the sequence of operations to construct graph $G'$ using at most $2k+O(1)$ labels, these being two labels for each $i\in[1,k]$ and a small number of ``work'' labels, along with a ``junk'' label to which we relabel the completed parts of the graph in order to reuse their work labels. We first introduce all vertices of sets $A_i$ and $B_i$ using one label per set and then all vertices of $W_i$ using one of the work labels. We then join $W_i$ to both $A_i,B_i$ and relabel $W_i$ to the junk label, thus being able to reuse its former work label for introduction of the following set $W_{i+1}$.
 
 Following this procedure for all $i\in[1,k]$, we turn to introduction of the quadruples of vertices for each edge $e=(v_i^l,v_j^h)\in E$. We introduce each of the 4 vertices $a_l^e,b_l^e,a_h^e,b_h^e$ using one label per vertex and join them to the appropriate sets $A_i,B_i,A_j,B_j$. We then introduce the $2d+2$ vertices of the OR gadget connecting $a_l^e,b_l^e$ using some work label and join them to both $a_l^e,b_l^e$. Relabelling the vertices of the OR gadget to the junk label allows us to reuse their previous work label for the remaining OR gadgets connecting the quadruple. Having thus completed the construction for edge $e$, we can relabel all vertices to the junk label and reuse their former work labels for the quadruple of the remaining edges.
 
 Concerning the cliques attached to each vertex to set their capacity, observe that we can introduce the vertex in question using some work label first, then introduce the vertices of the clique that should be adjacent to it using another work label, join them and then continue with the construction of the complete subgraph that forms the gadget using the same work label (see also Lemma \ref{lem:saturation}). Having thus constructed the clique for this vertex, we can relabel it to its ``proper'' label as described above in order to continue with our construction and also relabel all vertices of the clique to the junk label, thus being able to reuse the same work labels for the remaining cliques.
\end{proof}

\begin{theorem}\label{thm_whard_cw} \textsc{$d$-Orientable Deletion} is
W[1]-hard parameterized by the clique-width of the input graph. Furthermore, if
there exists an algorithm solving \dOr\ in time $n^{o(\cw)}$, then the ETH is
false. \end{theorem}
\begin{proof} Given some instance of \textsc{$k$-Multicolored Independent Set},
we construct an instance of \textsc{$d$-Orientable Deletion} with $d=n$, as
described above. Lemmas \ref{lem_whard_cw_FWD}, \ref{lem_whard_cw_BWD} show
correctness of the reduction, while Lemma \ref{lem_whard_cw_bound} gives the
bound on the clique-width of the constructed graph. For the running time bound
we recall that an algorithm for \textsc{$k$-Multicolored Independent Set}
running in time $n^{o(k)}$ would contradict the ETH. \end{proof}

\section*{Acknowledgement}
The very useful remarks and suggestions of the two anonymous referees are gratefully acknowledged.

\bibliography{orientable}

\end{document}